\documentclass[12pt,authoryear,review]{elsarticle}
\makeatletter

  \def\ps@pprintTitle{%
 \let\@oddhead\@empty
 \let\@evenhead\@empty
 \def\@oddfoot{\centerline{\thepage}}%
 \let\@evenfoot\@oddfoot}
\makeatother
\usepackage{amsmath,amssymb,amsfonts}
\usepackage{mathrsfs}
\usepackage{graphicx}
\usepackage{subfigure}
\usepackage[usenames,dvipsnames]{xcolor}
\usepackage{setspace}
\usepackage{ulem}
\usepackage{multirow}
\usepackage{lscape}
\usepackage{bigstrut}
\usepackage{enumitem}
\usepackage{dsfont}
\usepackage{stackengine}
\usepackage{color}
\usepackage[rightcaption]{sidecap}
\usepackage{epstopdf}

\newtheorem{theorem}{Theorem}
\newtheorem{lemma}[theorem]{Lemma}
\newtheorem{proposition}[theorem]{Proposition}

\newtheorem{definition}[theorem]{Definition}
\newenvironment{proof}[1][Proof]{\begin{trivlist}
\item[\hskip \labelsep {\bfseries #1}]}{\end{trivlist}}

\journal{TBA}

\usepackage[plainpages=true, pdfpagelabels]{hyperref}

\begin{document}
\begin{frontmatter}

\title{Insider Trading with Temporary Price Impact}

\author[author1]{Weston Barger}\ead{wdbarger@uw.edu}
\author[author2]{Ryan Donnelly}\ead{ryan.f.donnelly@kcl.ac.uk}

\address[author1] {University of Washington, Seattle, WA, United States}
\address[author2] {King's College London, London, UK}

\date{}

\begin{abstract}

	We model an informed agent with information about the future value of an asset trying to maximize profits when subjected to a transaction cost as well as a market maker tasked with setting fair transaction prices. In a single auction model, equilibrium is characterized by the unique root of a particular polynomial. Analysis of this polynomial with small levels of risk-aversion and transaction costs reveal a dimensionless parameter which captures several orders of asymptotic accuracy of the equilibrium behaviour. In a continuous time analogue of the single auction model, incorporation of a transaction costs allows the informed agent's optimal trading strategy to be obtained in feedback form. Linear equilibrium is characterized by the unique solution to a system of two ordinary differential equations, of which one is forward in time and one is backward. When transaction costs are in effect, the price set by the market maker in equilibrium is not fully revealing of the informed agent's private signal, leaving an information gap at the end of the trading interval. When considering vanishing transaction costs, the equilibrium trading strategy and pricing rules converge to their frictionless counterparts.

\end{abstract}
\begin{keyword}
	market microstructure, asymmetric information, price impact, transaction cost
\end{keyword}
\end{frontmatter}

\section{Introduction}

When traders place orders on a securities exchange, they face transaction frictions. Direct frictions include exchange and brokerage fees, but consumers of liquidity also experience indirect costs. Liquidity providers adjust their limit orders to reflect the information contained in incoming market orders by moving their price quotes in the direction of order flow. This adversely affects traders who take liquidity as their subsequent orders will be transacted at a less favourable price. Additionally, if an aggressive order is large enough then it consumes all of the liquidity at the best available price and the remainder of the order is executed at sequentially worse prices. This can be thought of as a transaction cost which is dependent on the size of the aggressive order and the state of the order book. %In this work, we consider an insider trading model in the spirit of \cite{kyle1985continuous} for which a trader with inside information is risk-averse and pays a transaction cost proportional to the size of her orders.

Risk-aversion and transaction costs have been previously studied in the insider trading literature. The authors \cite{holden_subrahmanyam_1994} extend the discrete-time model of \cite{kyle1985continuous} to include an exponentially risk-averse insider. Furthermore, their model allows for multiple insiders with the same level of risk-aversion who all receive identical information. \cite{subrahmanyam1998transaction} further extends the model of \cite{holden_subrahmanyam_1994} to include a quadratic transaction cost for a risk-averse insider in discrete-time. \cite{baruch2002insider} extends the continuous-time model which was first given by \cite{kyle1985continuous} and generalized by \cite{back1992insider} by including risk-aversion.

In this work, we model an exponentially risk-averse (or risk-neutral) insider who faces a transaction cost per share that is proportional to the size of the order. We first present a single-auction model and classify the unique linear equilibrium. We show that the market maker's equilibrium pricing rule corresponds to the unique positive root of a particular polynomial determined by the model parameters. An asymptotic expansion of the roots of the aforementioned polynomial is performed for small transaction cost and small risk-aversion, which allows us to examine the effects of the transaction cost relative to frictionless models.

We then formulate an analogous model in continuous-time and present a linear equilibrium classified by the solution to a forward-backward ordinary differential equation (FBODE). We show that the resulting FBODE has a unique solution and is explicitly solvable when the insider is risk-neutral. We then analyze the effects of varying the model parameters on equilibrium using numerical solutions of the associated FBODE. Although we cannot solve for equilibrium explicitly unless the insider is risk-neutral, we are able to make conclusions about the nature of equilibrium in certain limiting cases of the transaction cost parameter. In particular, when the transaction costs vanish the equilibrium trading and pricing rules converge to their frictionless counterparts. This result could be used to simplify the analysis of other similar asymmetric information models because it provides a family of feedback controls which converge to the equilibrium control in the frictionless case. 

The models formulated in this paper are also related to those of optimal execution literature. Often in that literature, the pressure on the asset price exerted by order flow is referred to as permanent price impact, while the immediate cost associated with market microstructure is referred to as temporary price impact. In their seminal work, the authors of \cite{almgren2001optimal} model permanent and temporary price impact by defining two distinct price processes: the midprice and the transaction price. The midprice is the midpoint between the best quoted bid and ask prices set by liquidity providers, and the transaction price is the average price per unit of asset at which the trader collects proceeds from trades. The authors of \cite{almgren2001optimal} model permanent impact by letting the drift of the midprice be an exogenous function of the trader's order flow. Temporary impact is modelled by defining the transaction price of trades to be equal to the midprice plus an exogenous function of the trader's order volume. 

Our model also includes midprice and transaction price processes. As distinct from \cite{almgren2001optimal}, we directly model price setting market makers which allows for the permanent price impact to be endogenous. However, we define a transaction price process that is analogous the model of \cite{almgren2001optimal} by explicitly introducing an exogenous transaction cost. To be consistent with the insider trading literature, we refer the permanent price impact effect simply as price impact and temporary price impact effect as transaction cost.

The continuous-time version of our model is a direct generalization of the continuous-time models of the aforementioned papers \cite{kyle1985continuous}, \cite{back1992insider} (when the insider's signal is Gaussian), and \cite{baruch2002insider} (with constant volatility of noise trading). As such, some qualitative features of equilibrium in these works also arise in the present paper, but there are also some notable differences. In particular, a key feature of many other models with asymmetric information is that the asset price is always fully revealing of the insider's signal at the end of the trading horizon (the insider always has incentive to exploit her informational advantage). This is not the case in our model when the transaction cost is non-zero. As a consequence, revelation of the insider's signal does contain information not already incorporated in the publicly available price.

The rest of the paper is organized as follows. In Section \ref{sec:singleAuction}, we develop a single-auction model, present the unique linear equilibrium, and analyze the effects of the transaction cost by performing an asymptotic expansion for small transaction cost. We shift our focus to continuous-time in Section \ref{sec:continuousTime} and begin by presenting a continuous-time model in Section \ref{sec:continuousTimeModel}. We begin Section \ref{sec:continuousTimeEquilibrium} by developing the mathematical machinery necessary the subsequent presentation of the linear equilibrium, and we finish the section by presenting a linear continuous-time equilibrium. In Section \ref{sec:parameterDependence}, we demonstrate the effects on the equilibrium of the previous section of varying the model parameters. In Section \ref{sec:limitingEquilibria}, we pay special attention to the transaction cost parameter by analyzing the limit of the equilibrium as this parameter tends to zero and infinity. Some concluding remarks are offered in Section \ref{sec:conclusion}.

\section{Single-Auction}
\label{sec:singleAuction}

In this section, we consider a single-auction market where the transaction price of the insider's trades incurs an additional cost per share which is linear with respect to the trade volume. After presenting the model which describes the dynamics of trade and the objective of the insider and market maker, we prove the existence of a unique equilibrium in this setting. We then investigate the effects of the transaction cost on the associated equilibrium.

The single-auction model in this section is similar to work contained in \cite{subrahmanyam1998transaction} for the case of a single agent. Though structurally similar, this previous work never considers a model in which risk-averse insiders interact with unpredictable noise traders. In that paper either the insider is risk-neutral or the volume traded by the noise traders is directly observed by the insider before submitting their own trade (in that case the stochasticity in the model comes from a random endowment to the insiders). Some of our results are analogous to \cite{subrahmanyam1998transaction}, but there are some distinctions, for example that our model guarantees equilibrium whereas the lack of noise traders can give rise to situations with no equilibrium (see Lemma 2 of  \cite{subrahmanyam1998transaction}). We include the single-auction results so that we may perform a more in-depth analysis of the equilibrium through an asymptotic expansion, and for the sake of completeness before investigating a continuous-time version of the model.

\subsection{Model}
\label{sec:singleAuctionModel}

In the spirit of \cite{kyle1985continuous}, we consider a single-auction on a market with one risky asset that is traded on an exchange with three types of traders: market makers who set the asset's midprice, an insider who has information about the future value of the asset, and noise traders. We let $v$ denote the ex-post liquidation value of the asset, and we assume that $v \sim \mathcal N (v_0, \Sigma_0^v)$. 

The insider receives the realization of $v$ before the auction takes place, but this information is unavailable to the public. Thus, she wishes to utilize her informational advantage by submitting an order of size $\Delta x$ in an auction. We assume that the number of noise traders on the exchange is large, and we denote the aggregate order of the noise traders by $\Delta z$, which we assume to be distributed as $\Delta z \sim \mathcal N (0, \sigma^2)$. The aggregate order of the insider and noise traders are submitted to the market maker who observes on the total quantity $\Delta y$, where
\begin{align}
	\Delta y &= \Delta x + \Delta z\,.\label{eq:delta.y}
\end{align}
The auction takes place in two phases. First, the insider and the noise traders submit orders to the exchange, and second, the market maker observes the aggregate order $\Delta y$ and sets the midprice $p$. 

We assume that the insider pays a transaction cost proportional to the size of her order so that the effective transaction price is
\begin{align}
\widehat p &= p + c \, \Delta x\,,
\end{align}
where $c > 0$ is a constant referred to as the transaction cost parameter. Note that as the insider sells shares of the asset her transaction price is lower than the midprice, and, conversely, as the insider purchases shares her transaction price is higher than the midprice.

The difference between the midprice, $p$, and the effective transaction price, $\widehat{p}$, could arise from one of many sources. The interpretation given in \cite{subrahmanyam1998transaction} is that of a transaction tax, possibly invoked upon the market by a regulator. A different interpretation could be that the effective transaction price is due to different preferences between a large number of market makers. In previous works, it is assumed that there is a very large number of perfectly competitive risk-neutral market makers which drives all of them to quote the same price. In reality, liquidity providers may set different prices than each other, giving rise to a demand structure depending on their aggregate quotes. We do not explicitly model this behaviour or interaction between market makers and instead capture this effect through the linear dependence of transaction price on trade volume.

Without a loss of generality, we assume that the insider holds no shares of the asset before the auction. The wealth of the insider after the trades are executed is thus
\begin{align}
	w &= (v - \widehat p)\,\Delta x\,.
\end{align}
The insider would like to maximize the utility of her expected wealth $w$. That is, the insider chooses $\Delta x$ to achieve
\begin{align}
	\max_{\Delta x} \mathbb{E}[U(w)\,|\,v]\,,\label{eq:singleAuctionUtilityMax}
\end{align}
where $U$ is the insider's utility function. The market maker is tasked with setting prices efficiently. That is, the market maker chooses the price $p$ such that 
\begin{align}
	p &= \mathbb{E}[v\,|\,\Delta y]\,.\label{eq:singleAuctionEffCond}
\end{align}

\subsection{Single-Auction Equilibrium}
\label{sec:singleAuctionEquil}

We begin this section by defining what it means for a pricing rule and a trading strategy to form an equilibrium. We then classify the unique linear equilibrium for an exponentially risk-averse (or risk-neutral) insider.

We assume that the market maker and insider choose pricing rules and trading strategies, respectively, as functions of information available to them during the auction. That is, the market maker chooses the price $p$ as a function of $\Delta y$, and the insider chooses her order size $\Delta x$ as a function of $v$. Let $P$ and $X$ be functions such that $p = P(\Delta y)$ and $\Delta x = X(v)$.

\begin{definition}\label{def:singleAuctionEquil}
	A \textbf{single-auction equilibrium} $(P, X)$ consists of a pricing rule $P$ and trading strategy $X$ such that 
	\begin{itemize}
		\item given the pricing rule $P$, the order $\Delta x$ given by the trading strategy $X(v)$ achieves the maximum in \eqref{eq:singleAuctionUtilityMax}, and 
		\item given the trading strategy $X$, the price $p$ given by the pricing rule $P(\Delta y)$ satisfies the efficiency condition \eqref{eq:singleAuctionEffCond}.
	\end{itemize}
	The pair $(P,X)$ is a \textbf{linear single-auction equilibrium} if both $P$ and $X$ are linear functions of their arguments.
\end{definition}

Suppose that $\mathscr P$ is the set of pricing rules $P \equiv P(\Delta y)$ and let $\mathscr P_0 \subset \mathscr P$ be the set of pricing rules for which there exists a corresponding trading strategy $X$ that satisfies \eqref{eq:singleAuctionUtilityMax}. Similarly, let $\mathscr X$ be the set of trading strategies $X \equiv X(v)$ and let $\mathscr X_0 \subset \mathscr X$ be the set of trading strategies for which there exists a pricing rule $P$ that satisfies \eqref{eq:singleAuctionEffCond}. The sets $\mathscr P_0$ and $\mathscr X_0$ induce mappings $\rho: \mathscr P_0 \to \mathscr X$ and $\xi: \mathscr X_0 \to \mathscr P$. A pricing rule $P$ and trading strategy $X$ form a single-auction equilibrium if $\rho(P) \in \mathscr X_0$, $\xi(X) \in \mathscr P_0$ and $(P,X) =  (\xi(X), \rho(P))$.

We will restrict our focus to the exponentially risk-averse (or risk-neutral) insider. For any constant $A \geq 0$, we define the utility function 
\begin{align}
	U(w) &= \left\{ \begin{array}{cr}
					w, & A = 0 \\
					- e^{- A w}, & A > 0
					\end{array}
	\right. \,, \label{eq:exponentialUtilityFunction}
\end{align}
and we refer to $A$ as the risk aversion parameter.

Before presenting the classification of linear equilibrium it is helpful to introduce the constants 
\begin{align}
	\lambda_K &= \frac 12 \sqrt{\frac{\Sigma_0^v}{\sigma^2}}, & \beta_K &= \frac{1}{2 \,\lambda_K}\,.\label{eq:kyleSingleAuctionLambda}
\end{align}

The constants $\lambda_K$ and $\beta_K$ correspond to the single-auction pricing rule and trading strategy, respectively, of \cite{kyle1985continuous}. In that work, the market maker's single-auction equilibrium pricing rule for a risk-neutral insider is $P_K( \Delta y ) = v_0 + \lambda_K \,\Delta y$, and the corresponding optimal trading strategy is $X_K(v) = \beta_K(v-v_0)$. We will write the equilibrium pricing rule and trading strategy of our model in terms of $\lambda_K$ which helps to illuminate the effect of the added transaction cost.

\begin{theorem}[Single Auction Equilibrium]	\label{thm:singleAuctionEquil}
	Choose $A \geq 0$ and let $U$, defined in \eqref{eq:exponentialUtilityFunction}, be the insider's utility function.	Then the unique linear single-auction equilibrium is given by
	\begin{align}
		P(\Delta y) &= v_0 + \lambda\,\Delta y\,,\\
		X(v) &= \beta\,(v-v_0)\,
	\end{align}
	where
	\begin{align}
		\beta &= \frac{	1 }{ 2\, (\lambda + c )+ A\, \sigma^2 \,\lambda^2 }\,, \label{eq:singleAuctionBeta}
	\end{align}
	and $\lambda$ is the unique positive root of the polynomial
	\begin{align}
		r(x) &=  A^2 \,\sigma^4 \,x^5 
					+ 4 \,A\, \sigma^2\, x^4
						+ 4 \,(1 + A\, c \,\sigma^2 )\,x^3
							+ 4\, (2\, c - A \,\sigma^2 \,\lambda_K^2)\, x^2
								+ 4\, (c^2 - \lambda_K^2 )\, x
									- 8 \,c\, \lambda_K^2\,. \label{eq:singleAuctionLambdaPoly}
	\end{align}
\end{theorem}
\begin{proof}
	For a proof see Section \ref{sec:pf_thm:singleAucionEquil} in the appendix.
\end{proof}

Intuitively, a risk-averse trader prefers to submit a smaller order compared to her risk-neutral counterpart because the inventory holdings are exposed to the randomness associated with the noise traders. For this reason one might expect that the insider submits smaller orders with as the value of $A$ is increased. Correspondingly, less information about the true value of the asset would be contained in the order signal $\Delta y$ received by the market maker, and she thus reduces the severity of her price adjustment. One might also expect that as the transaction cost parameter $c$ increases it becomes less worthwhile for the insider to submit large orders, regardless of her risk preference, thus leading to a smaller $\beta$ and $\lambda$. Both of these statements are indeed true as summarized by the following proposition.

\begin{proposition}[Single Auction Parameter Dependence]\label{prop:single_auction_dependence} Let $P(\Delta y) = v_0 + \lambda \Delta y$ and $X(v) = \beta(v-v_0)$ be the pricing rule and trading strategy forming the unique single-auction equilibrium given by Theorem \ref{thm:singleAuctionEquil}. If $A>0$ then
	\begin{align*}
		\frac{\partial \beta}{\partial c}  &< 0\,, & \frac{\partial \beta}{\partial A} & < 0\,,\\
		\frac{\partial \lambda}{\partial c}  &< 0\,, & \frac{\partial \lambda}{\partial A} & < 0\,.
	\end{align*}
\end{proposition}
\begin{proof}
	For a proof see Section \ref{sec:pf_prop:single_auction_dependence} in the appendix.
\end{proof}

The classification of equilibrium given by Theorem \ref{thm:singleAuctionEquil} is in terms of a root of a fifth degree polynomial. In general the quantities $\lambda$ and $\beta$ involved in the equilibrium will not have closed form expressions in terms of model parameters. However, we are able to find approximations to these quantities which hold when certain model parameters are small. These approximations are given in the following proposition.

\begin{proposition}[Single Auction Approximation]\label{prop:single_auction_approx} Let $P(\Delta y) = v_0 + \lambda \Delta y$ and $X(v) = \beta(v-v_0)$ be the pricing rule and trading strategy forming the unique single-auction equilibrium given by Theorem \ref{thm:singleAuctionEquil}, and define the dimensionless parameter
	\begin{align}
		\nu := \frac{\lambda_K\,\sigma^2\,A}{2} + \frac{c}{\lambda_K}\,.\label{eqn:nu_def}
	\end{align}
	Then the quantities $\lambda$ and $\beta$ admit the following approximations:
	\begin{align}
		\lambda &= \lambda_K\,\biggl(1 - \frac{1}{2}\,\nu^2 + \nu^3\biggr) + o(\nu^4)\,, \label{eqn:expansion_lambda}\\
		\beta &= \frac{1}{2\,\lambda_K}\,\biggl(1 - \nu + \frac{3}{2}\,\nu^2\biggr) + o(\nu^3)\,.\label{eqn:expansion_beta}
	\end{align}
\end{proposition}
\begin{proof}
	For a proof see Section \ref{sec:pf_prop:single_auction_approx} in the appendix.
\end{proof}

As should be expected, the approximations to $\lambda$ and $\beta$ both converge to their counterparts present in \cite{kyle1985continuous} as $A,c\rightarrow 0$, namely $\lambda_K$ and $1/(2\lambda_K)$. In addition, there are some other interesting observations to be made about this result.

The first and most significant observation is that both approximations depend on the dimensionless parameter $\nu$ as defined in \eqref{eqn:nu_def}, and not on the individual values of $A$ or $c$. A priori, there is no reason why this dimensionless parameter $\nu$ fully determines the approximations, and not the individual values of $A$ and $c$. The authors have verified that higher order approximations with respect to $A$ and $c$ fully depend on their individual values, and thus expanding with respect to the dimensionless parameter $\nu$ alone is not possible to higher order that what appears in the Proposition \ref{prop:single_auction_approx}.

The second observation stems again from the dimensionless parameter $\nu$, which means that both parameters have the same qualitative effect on equilibrium locally around $A = c = 0$. The magnitude of the effects, however, depend on other model parameters. For example, if $\lambda_K$ is large, then increasing $A$ will have more effect on equilibrium than increasing $c$\footnote{Recall that $\lambda_K$ depends on the other market parameters $\Sigma_0^v$ and $\sigma$, so this remark depends on changing $\lambda_K$ without changing $\sigma$}. This can be made sense of intuitively from the perspective of the insider. If Kyle's model is taken as a reference, then $\lambda_K\sigma^2$ represents a magnitude of price risk faced by the insider due to the random noise traders, and $\lambda_K$ on its own represents a cost of trading to the insider. If $\lambda_K$ is large, then the insider faces a large amount of risk and we would expect any increase in risk-aversion from zero to have a significant effect on equilibrium. Similarly, when $\lambda_K$ is large, increasing the value of $c$ from zero represents a relatively small change in the total trading cost to the insider and should have insignificant effect on equilibrium.

The third observation is that the approximation of $\lambda$ does not have any first order correction terms. This is an indication that locally around $A=c=0$, the equilibrium value of $\lambda$ is relatively robust with respect to changes in $A$ and $c$ because it is affected only at second order and higher. The expansion for $\beta$ however does have first order corrections with respect to both parameters, so the insider's trading strategy does not retain the same level of robustness to small changes of $A$ and $c$ from $0$.

To demonstrate the accuracy of these approximations, we plot both the exact equilibrium quantities and the approximations for various sets of parameters. For the sake of visual clarity we plot this comparison with respect to a single scale parameter denoted $\theta$ which acts as follows: we fix values of $A$ and $c$ and find the equilibrium quantities and the corresponding approximations after making the replacements
\begin{align*}
	A \mapsto \theta\,A\,,\quad c\mapsto\theta\,c\,.
\end{align*}
We then plot the results as a function of $\theta$ in order to observe convergence as $\theta\rightarrow 0$. The results are shown in Figure \ref{fig:approx} for three pairs of $A$ and $c$. Due to the scaling introduced by $\theta$, the actual values of $A$ and $c$ are of less significance than their ratio. The curves which are shown would not be altered by changing $A$ and $c$ such that their ratio is fixed after a reparameterization of $\theta$.

\begin{figure}
	\begin{center}
		{\includegraphics[trim=140 240 140 240, scale=0.55]{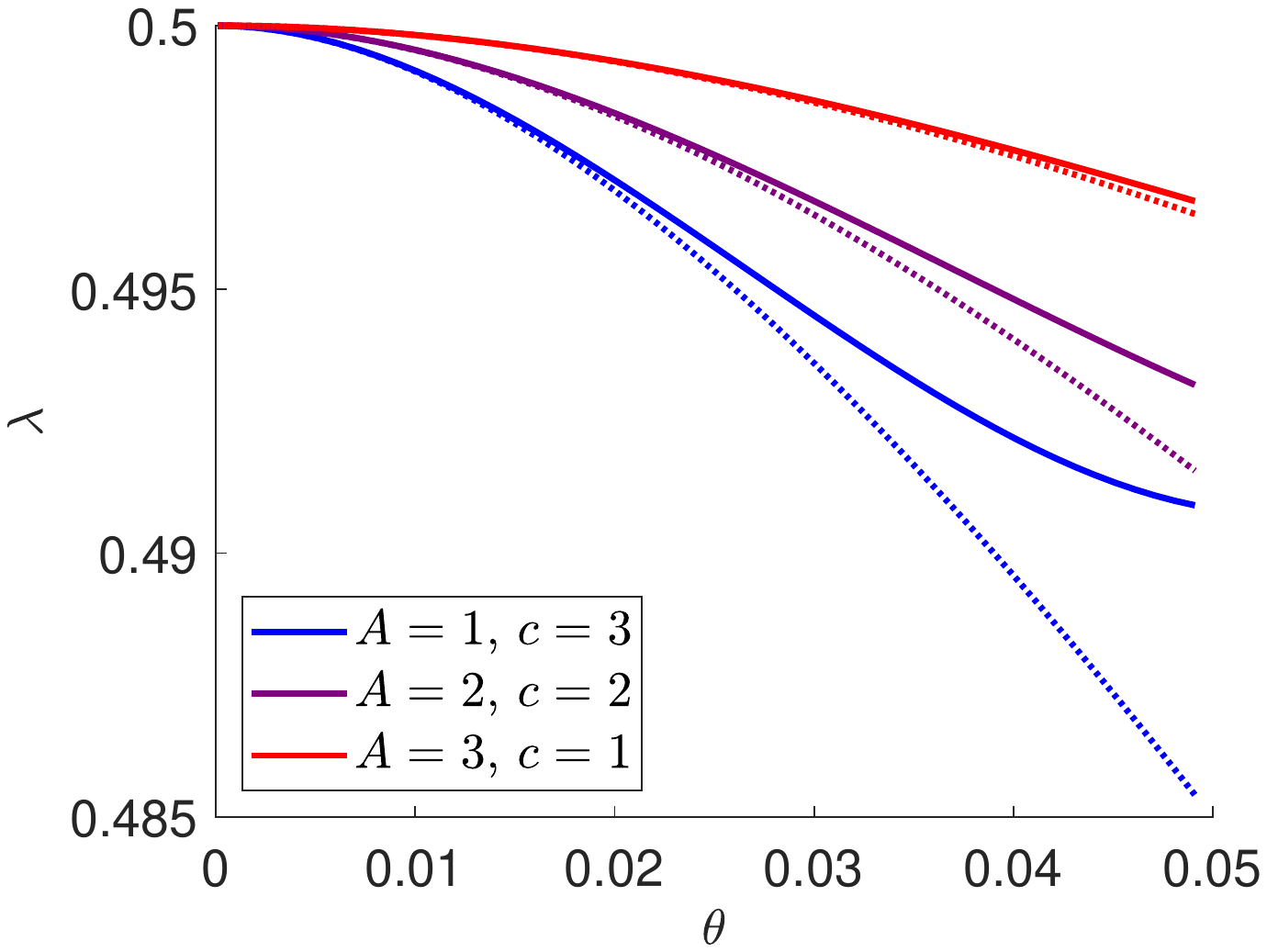}}\hspace{15mm}
		{\includegraphics[trim=140 240 140 240, scale=0.55]{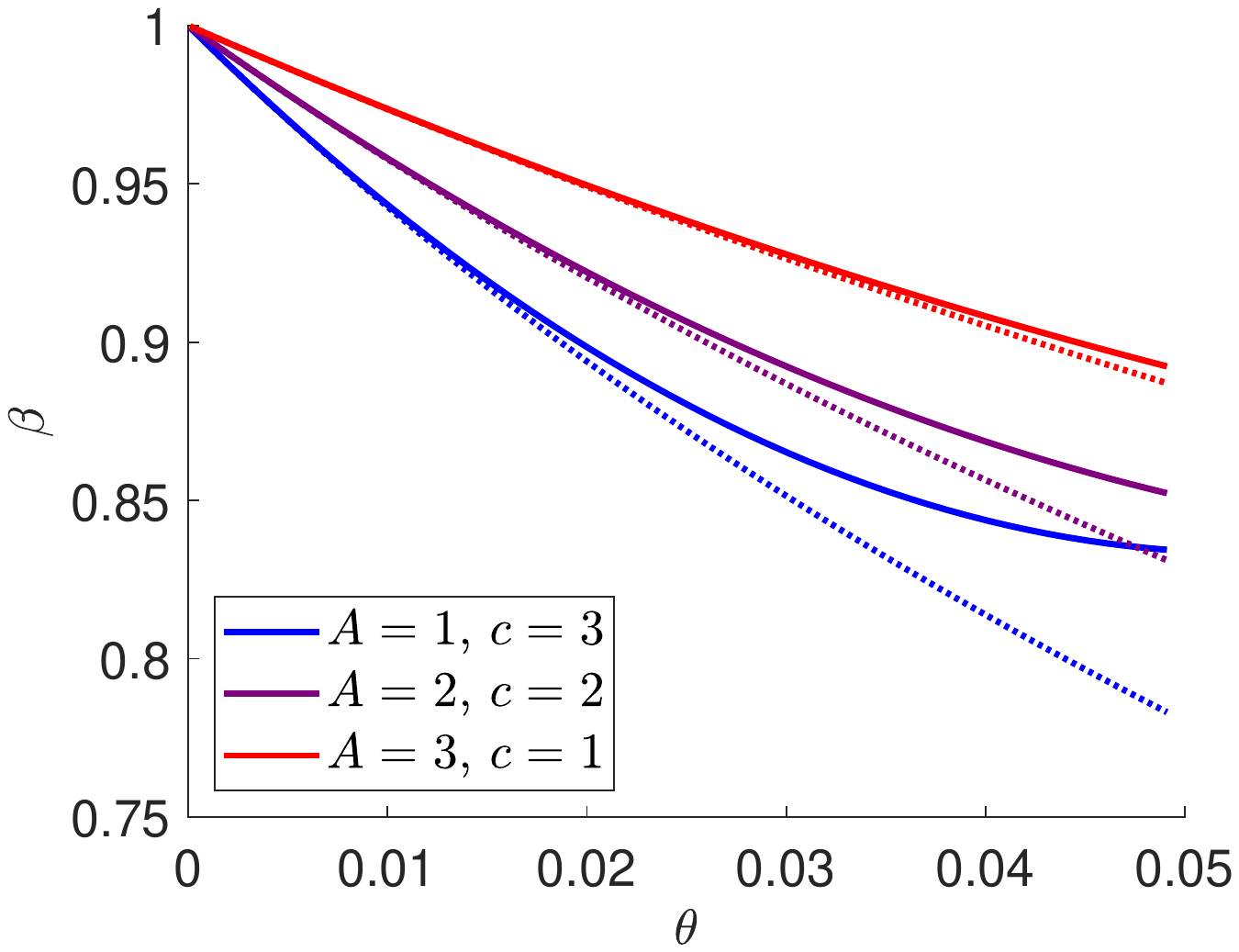}}
	\end{center}
	\vspace{-1em}
	\caption{Dotted curves show the exact equilibrium quantities for risk-aversion level $\theta A$ and transaction cost $\theta c$. Solid curves show the approximation of equilibrium quantities given by Proposition \ref{prop:single_auction_approx}. Other parameters are $\sigma = 1$ and $\Sigma_0^v = 1$. \label{fig:approx}}
\end{figure}

\section{Continuous-Time Auction}
\label{sec:continuousTime}

In this section, we present a continuous-time model which incorporates an analogous transaction cost to the previous section. After introducing the model in Section \ref{sec:continuousTimeModel}, we give the notion of market equilibrium which we consider. Some mathematical machinery is then established which is necessary for classifying equilibrium in continuous-time. When the insider is risk-neutral, we can solve for linear equilibrium in closed form. When they are risk-averse, the equilibrium is classified by a solution to an ODE, of which we prove there is always a solution. Finally, we analyze the dependence on the model parameters of the equilibrium solution, and give closed form expressions for the limits of the equilibrium as the transaction cost parameter tends to extreme values. 

\subsection{Continuous-Time Model}
\label{sec:continuousTimeModel}

We consider a similar setting to the previous section with three types of market participants, except orders are submitted and prices are set in continuous-time. Let the processes $P = (P_t)_{0\leq t\leq T}$ be the midprice of the asset, $X = (X_t)_{0 \leq t \leq T}$ be the insider's inventory, and  $Z = (Z_t)_{0 \leq t \leq T}$ be a standard Brownian motion. The cumulative orders submitted by the noise traders up to time $t$ is equal to $\sigma Z_t$, where $\sigma > 0$ is constant. We let $Y = (Y_t)_{0 \leq t \leq T}$ be the total number of shares submitted to the market up to time $t$. That is, 
\begin{align}
Y_t &= X_t - X_0 + \sigma\, Z_t\,. \label{eq:Yt}
\end{align}
As in the single-auction model, we assume that the market maker observes the aggregate order flow $Y$ but neither of the components $X$ or $Z$. We let the time $T$ value of the asset $v$ be normally distributed with mean $v_0$ and variance $\Sigma_0^v$, independent of the Brownian motion $Z$. The insider is privy to the realization of $v$ at the initial time $t = 0$.

It is helpful to explicitly denote the sets of information available to the various market participants at a given time. We define the filtrations $\mathcal{F}_t^Y$ and $\mathcal{F}_t^Z$ to be the filtrations generated by the processes $Y$ and $Z$, respectively. As the market maker observes only $Y$, we define $\mathcal{F}_t^M = \mathcal{F}_t^Y$ to be the market maker's information. The insider is aware of her trades, so she can back out the path of liquidity trades by observing historical prices (see \cite{back1992insider}), and the insider has received the realization of $v$ before trading begins. So we set the insider's information to be $\mathcal{F}_t^I = \sigma( \mathcal{F}_t^Z \cup \sigma ( v ))$. 

We assume that insider employs trading strategies that yield an absolutely continuous inventory process almost surely. To this end, we write 
\begin{align}
	d X_t &= \theta_t\, d t\,, 	\label{eq:dXt}
\end{align}
and analogously to the single-auction model, we define the insider's transaction price $\widehat{P} = (\widehat{P}_t)_{0 \leq t \leq T}$ to be 
\begin{align}
	\widehat{P}_t = P_t + c \, \theta_t\,, \label{eq:transactionPrice}
\end{align}
where $c$ is a positive constant. 

The assumption that the insider's inventory path is absolutely continuous may seem restrictive. In most other model formulations with asymmetric information, the insider may trade according to any process which is adapted to the filtration $\mathcal{F}_t^I$, including processes with discontinuities or diffusive components. However, it is generally the case that in equilibrium the insider's inventory is absolutely continuous. We do not attempt to model the transaction cost for trading strategies which are not absolutely continuous. There are many works in the context of portfolio optimization which include transaction costs of various forms (among many others, see \cite{magill1976portfolio}, \cite{davis1990portfolio}, and \cite{muhle2017primer}), but it is unlikely that the optimal strategies which result in those models would be consistent with the notion of equilibrium that we consider.

At the end of the trading horizon, the insider's wealth is $X_T \, v$ minus the cost of trading throughout the period. Thus, the terminal wealth $W_T$ is given by
\begin{align}
	W_T &= X_T \,v - \int_0^T \widehat{P}_s\, \theta_s\, d s\,.\label{eq:terminalWealth}
\end{align}
The insider chooses a strategy to maximize her expected utility of terminal wealth. That is, $\theta$ is chosen to achieve 
\begin{align}
	\sup_{\theta \in \mathcal A} \mathbb{E}[ U(W_T) \,|\,\mathcal{F}_0^I ]\,,\label{eq:generalOptimalUtility}
\end{align}
where $U$ is an increasing, concave function and the set of admissible strategies is
\begin{align}
\mathcal A = \{ \theta \mid \theta \text{ is } \mathcal{F}^I\text{-predicable and } \mathbb{E}\biggl[\int_{0}^T \theta^2_t \, dt\biggr] < \infty\}\,.
\end{align}
The market maker is tasked with setting the price of the asset efficiently at all times $0 \leq t \leq T$ and so should choose the price according to
\begin{align}
	P_t &= \mathbb{E} [ v \,|\,\mathcal{F}_t^M]\,. \label{eq:efficiencyCondition}
\end{align}

We now define the concept of equilibrium in continuous-time. This concept is analogous to the single-auction equilibrium concept of Section \ref{sec:singleAuctionEquil}.

\begin{definition}	\label{def:continuousTimeEquilibrium}
	A \textbf{continuous-time equilibrium} $(P, X)$ consists of a price process $P$ and an inventory process $X$ such that 
	\begin{itemize}
		\item given a price process $P$, the inventory process $X$ achieves the supremum in \eqref{eq:generalOptimalUtility}, and 
		\item given an inventory process $X$, the price process $P$ is efficient, i.e. $P$ satisfies \eqref{eq:efficiencyCondition}.
	\end{itemize}
\end{definition}

\subsection{Continuous-Time Equilibrium}
\label{sec:continuousTimeEquilibrium}

The goal of this section is to classify a linear equilibrium in continuous-time for an exponentially risk-averse or risk-neutral insider. We will restrict our consideration to pricing rules and trading strategies that have a form which are analogous to those of the single-auction equilibrium presented in Theorem \ref{thm:singleAuctionEquil}. Namely, increments of the price set by the market maker will be linear with respect to increments of trade volume, and the insider's trading strategy will be linear with respect to $v-P_t$. Before stating the classification of equilibrium in continuous-time we need to develop some additional mathematical machinery. These are contained in Lemmas \ref{lem:optimalTradingStrategy}, \ref{lem:filtering}, and \ref{lem:genExistAndUnique}.

Using expressions \eqref{eq:dXt}, \eqref{eq:transactionPrice}, and \eqref{eq:terminalWealth}, we rewrite the expression for terminal wealth $W_T$ as
\begin{align}
	W_T &= \int_0^T (v - P_s - c\, \theta_s )\,\theta_s\, ds\,,
\end{align}
where we have taken $X_0 = 0$ for simplicity. The computations that follow can be carried out with $X_0 \neq 0$, but this choice will not affect the insider's optimal trading strategy (it will have only a minor effect on the insider's value function).

We will develop optimal insider trading strategies by first fixing the market maker's pricing rule and then solving the Hamilton-Jacobi-Bellman (HJB) equation associated with the optimization problem \eqref{eq:generalOptimalUtility}. We will then subsequently verify that the solution is optimal. This motivates us to introduce the dynamic version of \ref{eq:generalOptimalUtility} as
\begin{align}
	H(t,P) &= \sup_{\theta \in \mathcal A} \mathbb{E}\biggl[ U\biggl( \int_t^T (v -  P_s - c \, \theta_s )\, \theta_s \, ds \biggr) \, \biggl|\,\mathcal{F}_t^I\biggr]\,. \label{eq:generalValueFunction}
\end{align} 
For a given time $t \in [0,T]$ and midprice $P_t$, the quantity $H(t,P_t)$ gives the optimal expected utility that the insider can achieve by trading during the time interval $[t,T]$. We refer to $H$ as the insider's value function. 

\begin{lemma}[Insider's Value Function and Optimal Strategy] \label{lem:optimalTradingStrategy}
		Let $A\geq 0$ and let the utility function $U$ be as defined in \eqref{eq:exponentialUtilityFunction}. Let $\lambda$ be a positive, bounded, deterministic function such that the Riccati differential equation
		\begin{align}
			\frac{dh(t)}{dt} &= - \frac{ (1 - 2\, A\, c\, \sigma^2 ) \lambda^2(t) }{ c }\, h^2(t) + \frac{ \lambda(t) }{ c }\,h(t) - \frac{ 1 }{ 4\, c}\,, &  h(T) &= 0\,,	\label{ode:riccatiODE}
		\end{align}
		has a global solution $h:[0,T]\rightarrow\mathbb{R}$. Suppose the midprice process $P$ is given by 
		\begin{align}
			P_t &= v_0 + \int_0^t \lambda(s)\, dY_s,
		\end{align}
		where $Y$ is given in \eqref{eq:Yt}. Then the insider's value function \eqref{eq:generalValueFunction} is given by
		\begin{align}
			H(t,P) &= \left\{\begin{array}{lr}- \exp\left\{- A \biggl( (v - P )^2 \, h(t) + \sigma^2 \, \int_t^T \lambda^2(s)\, h(s)\,ds \biggr) \right\}\,, & A > 0\,,\\
						(v - P )^2 \, h(t) + \sigma^2 \, \int_t^T \lambda^2(s)\, h(s)\,ds\,,  & A = 0\,,\end{array}\right.	\label{eq:valueFunction}
		\end{align}
		and the optimal strategy in feedback form is given by
		\begin{align}
			\theta_t^* &= \beta(t)\, (v - P_t), & \beta(t) &= \frac{ 1 - 2 \,\lambda(t)\, h(t) }{ 2 \, c }\,.	\label{eq:optimalTradingStrategy}
		\end{align}
\end{lemma}
\begin{proof}
	For a proof see Section \ref{proof:optimalTradingStrategy} in the appendix.
\end{proof}

In the statement of Lemma \ref{lem:optimalTradingStrategy} we have eliminated consideration of many arbitrary pricing rules by considering only functions $\lambda$ for which there is a solution to the ODE \eqref{ode:riccatiODE}. This is necessary in establishing this lemma in generality, as it is easy to find\footnote{If $\lambda_0$ is relatively large and decreases linearly to $\lambda_T$ which is relatively small, then the ODE will only have a solution on an interval of the form $(T^*,T]$ where $T^*>0$}  functions $\lambda$ for which this ODE does not have a solution defined on all of $[0,T]$. This is because for some pricing rules there may be incentive for the insider to acquire arbitrarily large inventory positions which will push prices in her favour, and then liquidating the acquired position when price impact is smaller, thereby making an unbounded profit. We make more comments on how we avoid this issue after we present the continuous-time equilibrium.

We use an HJB approach to prove this lemma, and the elimination of many pricing rules from consideration is analogous to a similar result in the standard Kyle model. If the HJB approach is taken in the Kyle model, then it is straightforward to show that there is no solution to the HJB equation if $\lambda$ is not constant in the risk-neutral case, or if $\lambda$ doesn't have very specific dynamics in the risk-averse case. Thus, acquiring an optimal trading strategy which is admissible requires the discarding of many pricing rules. On the other hand, using the HJB approach in our setting with transaction cost does have a significant difference compared to the standard Kyle model. Due to the quadratic nature of the performance criteria, we are immediately provided with the feedback form of the optimal trading strategy. In addition, the linear performance criteria in the Kyle model means that the HJB equation reduces to a system of PDE's rather than a single equation as in our model. The addition of the transaction cost makes many of the mathematical components of our problem more straightforward.

In the next Lemma we provide the details of the market maker's task of setting efficient prices. Recall that in the proof of Theorem \ref{thm:singleAuctionEquil} the pricing rule was shown to be efficient by applying the projection theorem for normal random variables directly to the expectation $\mathbb{E}[v\,|\,\Delta y]$, where we had assumed that the insider's trading strategy was linear in $v$. Our approach in continuous-time is analogous, but we need a generalization of the projection theorem to continuous-time. In following Lemma, we assume the insider follows a linear trading strategy and apply optimal filtering theory to compute $\mathbb{E}[ v \,|\, \mathcal{F}_t^M]$. 

\begin{lemma}[Market Maker's Efficient Pricing]\label{lem:filtering}
	Suppose that insider's inventory process is specified by
	\begin{align}
		dX_t &= \beta(t)\, (v - P_t )\, dt\,,
	\end{align}
	for a deterministic function $\beta$. Then the process $P$ specified by the dynamics
	\begin{align}
	dP_t &= \lambda(t) \, dY_t, & P_0 &= v_0\,,
	\end{align}
	where $Y$ is given by \eqref{eq:Yt} and
	\begin{align}
		\lambda(t) &= \frac{\beta(t) \Sigma(t)}{\sigma^2}\,, & \frac{d\Sigma(t)}{dt} &= - \sigma^2 \lambda^2(t)\,, & \Sigma(0) &= \Sigma_0^v\,, \label{eqn:lambda_Sigma}
	\end{align}
	satisfies the efficiency condition \eqref{eq:efficiencyCondition}. Furthermore, the function $\Sigma$ is equal to the posterior variance of $v$ given the market maker's information:
	\begin{align}
		\Sigma(t) &= \mathbb{E}[(v-P_t)^2|\mathcal{F}_t^M]\,.
	\end{align}
\end{lemma}

\begin{proof}
	The result is a direct application of \cite[Theorem 12.1]{liptsershiryaev2001statistics}. \qed
\end{proof}

The previous two lemmas are analogous to the steps in the single-auction case. The reader will recall that the proof of Theorem \ref{thm:singleAuctionEquil} was done in two steps. The first step was the computation of the insider's optimal strategy for a fixed, linear pricing rule. In discrete time, this came down to solving an algebraic equation, of which Lemma \ref{lem:optimalTradingStrategy} is the continuous-time analogue. The second step was the computation of an efficient pricing rule for a fixed, linear trading strategy, and Lemma \ref{lem:filtering} is the corresponding result in continuous-time.

We reduced the computation of the linear single-auction equilibrium to coupled algebraic equations, which were subsequently reduced to a single algebraic equation. In continuous-time, we reduce the computation of an equilibrium to two coupled, nonlinear ODEs for which one ODE is prescribed an initial condition and the other is prescribed a terminal condition. This motivates the following definition. 

\begin{definition}
	Let $T,\xi_0, \xi_T \in \mathbb{R}$ be constants with  $T > 0$. Let $x:\mathbb{R}\rightarrow\mathbb{R}^2$ and $F: \mathbb{R}^2 \rightarrow \mathbb{R}^2$. A \textbf{forward-backward ordinary differential equation (FBODE)} is a system of the form
	\begin{align}
	\frac{dx(t)}{dt} &= F(x(t)), & 	\left\{ \begin{array}{l}
											x_1(0) =  \xi_0\\
											x_2(T) = \xi_T
											\end{array}
											\right. \,.
	\end{align}
\end{definition}

In the following lemma, we state the FBODE that will appear in our continuous-time equilibrium and prove the existence and uniqueness of a solution.
\begin{lemma}[Solution to FBODE]\label{lem:genExistAndUnique}
	Let $A \geq 0$, $c\geq 0$, and $\Sigma^v_0 > 0$. For $x \in \mathbb{R}^2$, let us define the function $F: \mathbb{R}^2 \rightarrow \mathbb{R}^2$ as 
	\begin{align}
	F(x)  &=  \begin{pmatrix}
	F_1(x) \\ F_2(x)
	\end{pmatrix}, & F_1(x) &= - \frac{ \sigma^2 \, x_1^2 }{ 4 \,( c\, \sigma^2 + x_2 )^2 }\,, & F_2(x) &= -\frac{\sigma^2\,x_1\,(c\,\sigma^2 + x_2 - 2\,A\,x_2^2)}{4\,(c\,\sigma^2 + x_2)^2}\,. \label{eq:F}
	\end{align}
	Then the FBODE
	\begin{align}
	\frac{dx(t)}{ dt } &= F(x(t))\,, & \left\{ \begin{array}{l}
	x_1(0) =  \Sigma^v_0\\
	x_2(T) = 0
	\end{array}
	\right. \,, \label{ode:FBODE}
	\end{align}
	has a solution. If $c>0$ then this solution is unique and it satisfies $x_1(t) > 0$ for all $t \in [0,T]$ and $x_2(t) > 0$ for all $t \in [0,T)$. If $c=0$ then the solution is unique if we impose $x_2(t) > 0$ for all $t\in[0,T)$. In this case the solution satisfies $x_1(t) > 0$ for all $t\in[0,T)$ and $x_1(T) = 0$.
	
	In addition, when $A=0$ the solution is given by
	\begin{align}
		x_1(t) &= 2\, c\, \lambda_0\, \sigma^2  + \lambda_0^2\, \sigma^2\,  (T-t), & x_2(t) &= \frac{ T-t }{ 2\, \lambda_0\,(T-t) + 4\, c }\,, & \lambda_0 &= \sqrt{\frac{\Sigma^v_0}{\sigma^2\,T} + \frac{c^2}{T^2}} - \frac {c}{T}\,, \label{eq:odeAlphaZeroSol}
	\end{align}
	and when $c=0$ the solution which satisfies $x_2(t)>0$ for $t\in[0,T)$ is given by
	\begin{align}
		x_1(t) &= \frac{4\,\sigma^2\,\Lambda_K^2\,(T-t)}{(A\,\Sigma^v_0+2\,S)\,(A\,\Sigma^v_0\frac{2\,t-T}{T} + 2\,S)}\,, & x_2(t) &= \frac{\sigma^2\,\Lambda_K^2\,(T-t)}{A\,\Sigma^v_0+2\,S}\,, & \Lambda_K &= \sqrt{\frac{\Sigma^v_0}{\sigma^2\,T}}\,, & S &= \sqrt{\biggr(\frac{A\,\Sigma^v_0}{2}\biggl)^2 + \Lambda_K^2}\,.\label{eq:odeCZeroSol}
	\end{align}
\end{lemma}
\begin{proof}
	For a proof see Section \ref{proof:genExistAndUnique} in the appendix.
\end{proof}

Note that the results of Lemma \ref{lem:genExistAndUnique} include the possibility of $c=0$ even though our model specifies that $c$ is positive. The fact that the FBODE in Lemma \ref{lem:genExistAndUnique} has a desired solution for $c=0$ will be useful in proving the nature of equilibrium in the limit $c\rightarrow 0$, and so we include these results.

It is also useful to mention which equilibrium quantities are represented by the functions $x_1$ and $x_2$ as solutions to the FBODE. Taking $x_1$ and $x_2$ the solution when $c>0$, we will construct equilibrium by letting $\Sigma(t) = x_1(t)$ and $h(t) = x_2(t)/x_1(t)$, which is well defined due to the result of Lemma \ref{lem:genExistAndUnique} that $x_1(t)>0$ for all $t\in[0,T]$.

Before giving a continuous-time equilibrium it is useful to define the constant
\begin{align}
\Lambda_K &= \sqrt{\frac{\Sigma_0^v}{\sigma^2\,T}}\,. \label{eq:kyleContinuousAuctionLambda}
\end{align}
The constant $\Lambda_K$ is the continuous-time pricing rule of \cite{kyle1985continuous}. That is, when a risk-neutral insider faces a risk-neutral market maker on a frictionless exchange, the processes specified by the dynamics
\begin{subequations}
	\label{eq:kyleEquilibrium}
	\begin{align}
	dP_t &= \Lambda_K\, dY_t\,, & P_0 &= v_0\,, \label{eq:kyleP}\\
	dX_t &= \frac{1}{\Lambda_K\,(T-t)}\, (v - P_t)\, dt\,, & X_0 &= 0\,, 	\label{eq:kyleX}
	\end{align}
\end{subequations}
form an equilibrium. Writing the risk-neutral continuous-time equilibrium in terms of $\Lambda_K$ will help illuminate the effect that transaction costs have on the market participant's respective strategies.

\begin{theorem}[Continuous-Time Equilibrium]	\label{thm:continuousTimeEquilibrium}
	Let $U$ be as given in \eqref{eq:exponentialUtilityFunction} for some $A\geq 0$, and let $c>0$. Let $x = (x_1,x_2)$ be the unique solution to the FBODE \eqref{ode:FBODE}, and let $\Sigma(t)=x_1(t)$ and $h(t)=x_2(t)/x_1(t)$. Then the midprice process $P$ and inventory process $X$ specified by 
	\begin{align}
		dP_t &= \lambda (t)\, dY_t\,, & P_0 &= v_0\,, \label{eq:continuousTimeEquilPriceProcess}\\
		dX_t &= \beta(t)\,( v - P_t )\,dt\,, & X_0 &= 0\,, \label{eq:continuousTimeEquilInventoryProcess}
	\end{align}
	where $Y$ is given in \eqref{eq:Yt} and where
	\begin{subequations}
		\label{eqs:equilibriumBetaAndLambda}
		\begin{align}
		\beta(t) &= \frac{ \sigma^2	}{ 2 \,(c\, \sigma^2 + \Sigma(t)\,h(t) ) }\,, \label{eq:equilibriumBeta}\\ 
		\lambda(t) &= \frac{ \Sigma(t) }{ 2\,( c\, \sigma^2 + \Sigma(t)\,h(t) )	}\,,\label{eq:equilibriumLambda}
		\end{align}
	\end{subequations}
	form a continuous-time equilibrium. Furthermore, $\Sigma(t) = \mathbb{E}[ (v-P_t)^2\, |\,\mathcal{F}_t^M]$ and the value function $H$ is given by \eqref{eq:valueFunction}. When $A = 0$, $\lambda(t) \equiv \lambda$ is a constant and
	\begin{subequations}
		\label{eqs:riskNeutralEquilibriumBetaAndLambda}
		\begin{align}
			\beta(t) &= \frac{ 1 }{ \lambda\, (T-t) + 2\, c	}\,, \label{eq:riskNeutralEquilibriumBeta}\\
			\lambda &= \sqrt{ \Lambda_K^2 + \frac{c^2}{T^2}} - \frac {c}{ T }\,. \label{eq:riskNeutralEquilibriumLambda}
		\end{align}
	\end{subequations}
\end{theorem}
\begin{proof}
	For a proof see Section \ref{proof:continuousTimeEquilibrium} in the appendix.
\end{proof}

Recall that in Lemma \ref{lem:optimalTradingStrategy} we removed from consideration many pricing rules in which an associated ODE did not have a solution, as failure to do so would result in the ODE solutions blowing up within the interval $[0,T]$. Our method of classifying equilibrium in Theorem \ref{thm:continuousTimeEquilibrium} avoids this issue by taking the function $h$ as given (in terms of the unique solution to an FBODE which does not blow up), and then specifying the pricing rule $\lambda$ which yields $h$ as the solution to the ODE \eqref{ode:riccatiODE}.

%%%%REMARK
Definition \ref{def:continuousTimeEquilibrium} defines what it means for processes $(P,X)$ to form a continuous-time equilibrium. In Theorem \ref{thm:continuousTimeEquilibrium}, we gave deterministic functions $\beta$ and $\lambda$ such that the processes defined by $dP_t = \lambda(t)\,dY_t$ and $dX_t = \beta(t)\,(v -P_t)dt$ with $(P_0, X_0) = (v_0, 0)$ form a continuous-time equilibrium. While the functions $(\beta, \lambda)$ do not themselves form an equilibrium, the functions $(\beta, \lambda)$ correspond directly to a linear equilibrium $(P,X)$. In the sequel, we will refer to linear equilibria $(P, X)$ using $(\beta, \lambda)$, and we refer to $\beta$ as an \textit{equilibrium trading rule} and $\lambda$ as an \textit{equilibrium pricing rule}.
%%%%END REMARK

We will now discuss and interpret some characteristics of the continuous-time equilibrium of Theorem \ref{thm:continuousTimeEquilibrium}.

\begin{proposition}[Trading and Pricing Rule Monotonicity]\label{prop:compareStat}
	Let $\beta$ and $\lambda$ be the trading and pricing rules, respectively, corresponding to the continuous-time equilibrium of Theorem \ref{thm:continuousTimeEquilibrium}. Then,
	\begin{enumerate}
		\item $\beta$ is an increasing function of $t$, 
		\item if $A > 0$ then $\lambda$ is a decreasing function of $t$.
	\end{enumerate}
\end{proposition}
\begin{proof}
	For a proof see Section \ref{proof:compareStat} in the appendix.
\end{proof}

In equilibrium, $\beta$ is increasing in time and $\beta(T) = 1/ 2 \,c$, so
\begin{align}
\label{eq:betaBounds}
	0 & <\beta(t) < \frac 1{2\,c}, & t&\in [0,T)\,.
\end{align}
To gain some intuition as to why the trading rule is bounded by $1/2c$, let $b(t)$ be an arbitrary trading rule and consider the expression for terminal wealth
\begin{align}
	W_T &= \int_{0}^T b(t)\, ( 1 - c\, b(t))\, (v - P_t)^2\, dt\,.\label{eq:linearWealth}
\end{align}
Thus, at time $t$ the insider gains wealth at a rate equal to $b(t)\, ( 1 - c\, b(t))\, (v - P_t)^2$. This expression is positive for $0<b(t)<1/c$ and maximized at $b(t)=1/2c$. However, in choosing $b(t)=1/2c$ price impact effects are also maximized and this lowers the potential gain of wealth at future times. The optimal balance of instantaneous and future gains is therefore achieved for trading rules which are bounded by $1/2c$.

In the absence of transaction costs, the authors of \cite{kyle1985continuous} and \cite{baruch2002insider} show that in equilibrium both the risk-neutral and exponentially risk-averse insider, respectively, take trading rules $\beta$ that blow up as $t$ approaches $T$ so as to force the midprice to $v$, resulting in $\Sigma(T) = 0$. For a non-zero transaction cost $c$, the discussion above shows that an insider of any risk tolerance never wishes to choose $\beta > 1 / 2c$. In this setting, the insider does not reveal the true value of the asset to the market maker by the end of the trading period i.e. $\Sigma(T) > 0$.

\section{Parameter Dependence}
\label{sec:parameterDependence}

In this section, we discuss the dependence of the equilibrium of Theorem \ref{thm:continuousTimeEquilibrium} on model parameters. First we summarize the effects of varying model parameters on the continuous-time equilibrium by numerically solving the associated FBODE. Then we study the limits of the equilibrium rules as the transaction cost parameters are taken to extreme values.

In analyzing the effects of varying the model parameters on the resulting equilibria, we consider the solution, $x$, to the FBODE \eqref{ode:FBODE} to be a function of the underlying parameters $x(t) \equiv x(t;\Theta)$, where $\Theta$ is a collection of parameters of interest. An ideal approach would then be to differentiate the FBODE \eqref{ode:FBODE} with respect to $\Theta_i$ to get dynamics of $\partial x/\partial \Theta_i$. Unfortunately the non-linearities make this approach highly intractable, and so we resort to demonstrating these dependencies by numerically solving the FBODE for various sets of the underlying parameters.

In Figure \ref{fig:dependence_c} we show the effect of varying the transaction cost parameter $c$. As would be expected of the trading rule $\beta$, when there are larger transaction costs the insider trades less aggressively as demonstrated in the left panel. Consequently, the price impact $\lambda$ shown in the middle panel decreases for larger transaction costs because net order flow contains less information. In addition the variance of the market maker's estimate decreases more slowly and price discovery takes comparatively more time, as seen in the right panel.

\begin{figure}
	\begin{center}
		{\includegraphics[trim=140 240 140 240, scale=0.45]{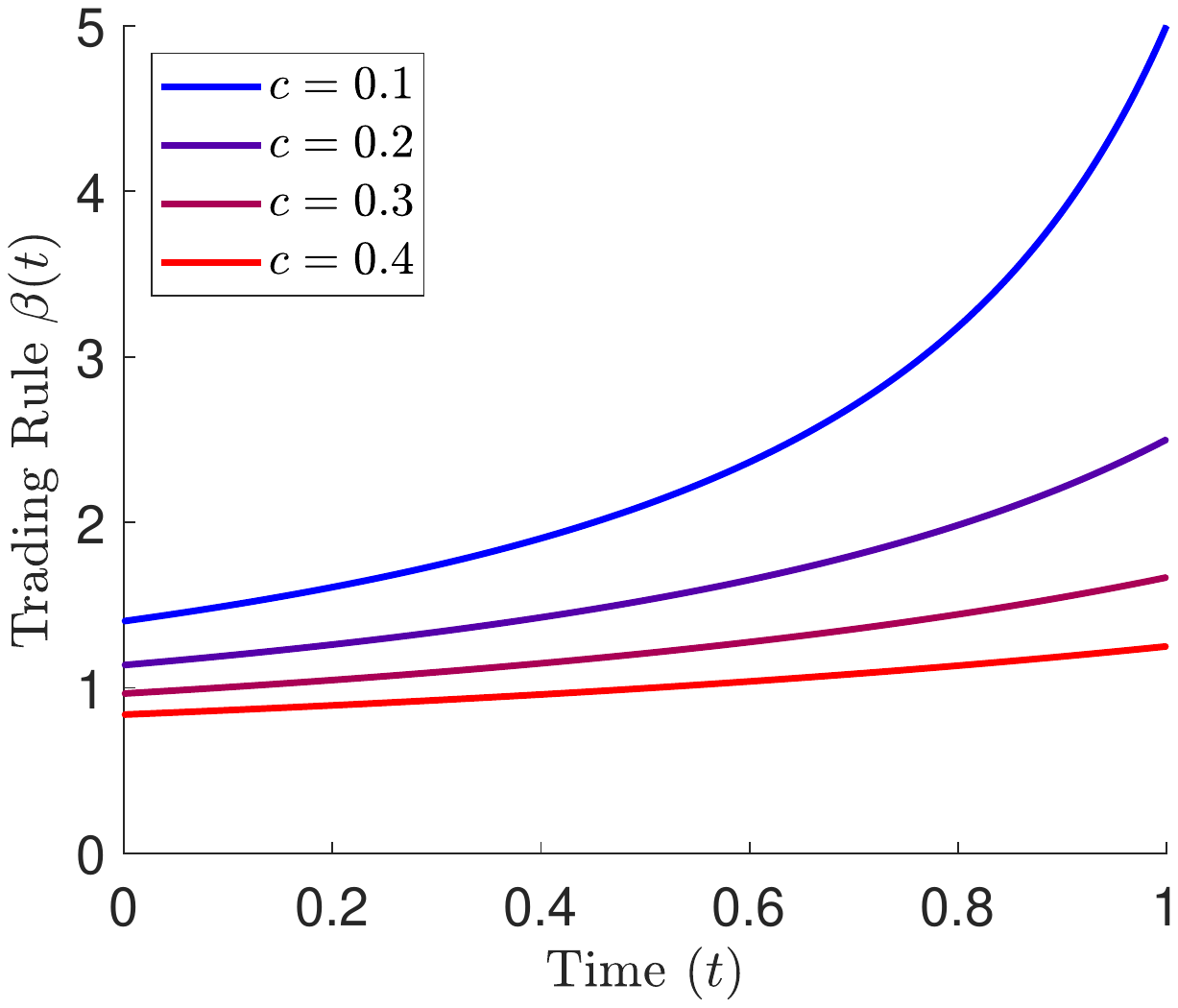}}\hspace{8mm}
		{\includegraphics[trim=140 240 140 240, scale=0.45]{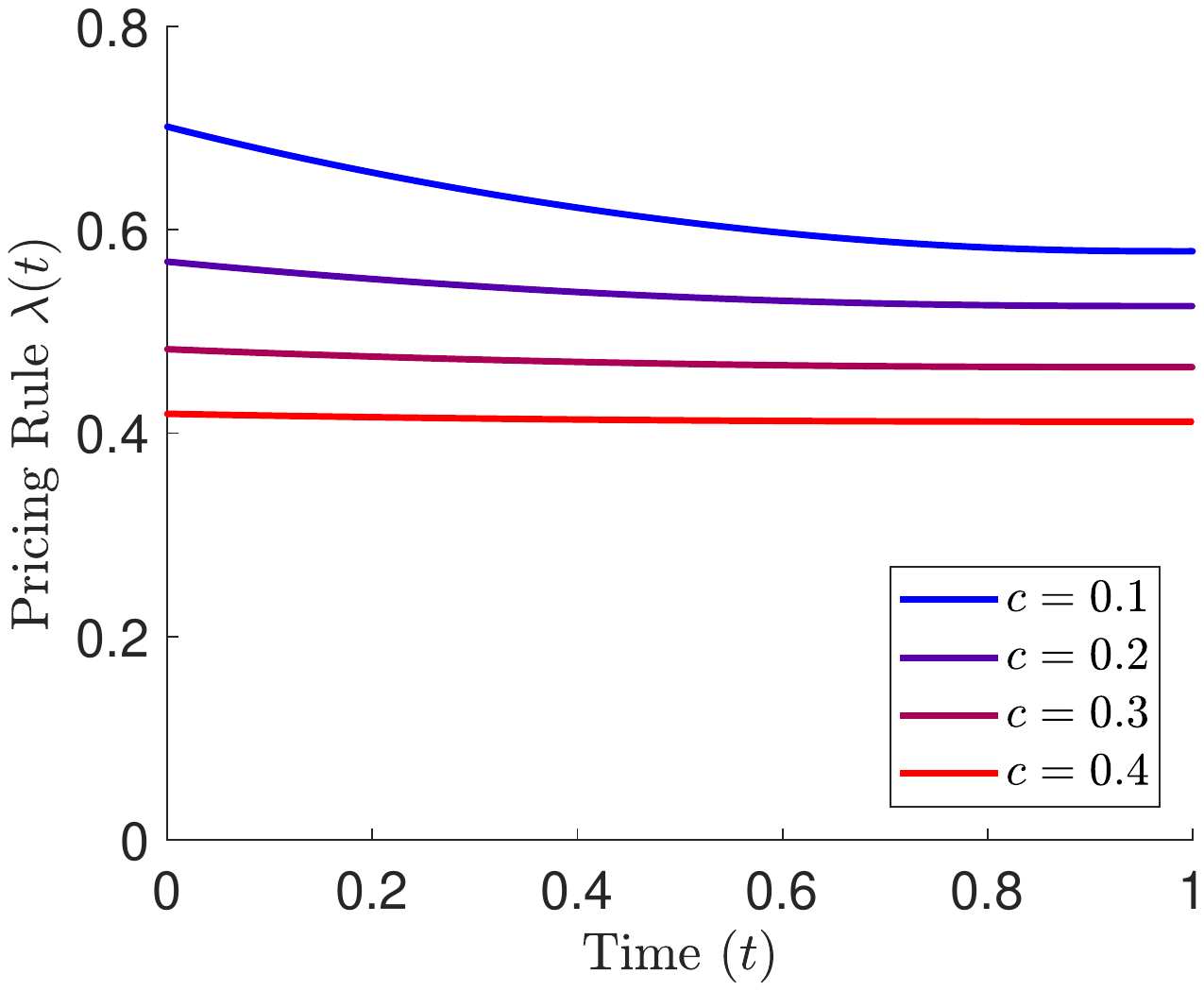}}\hspace{8mm}
		{\includegraphics[trim=140 240 140 240, scale=0.45]{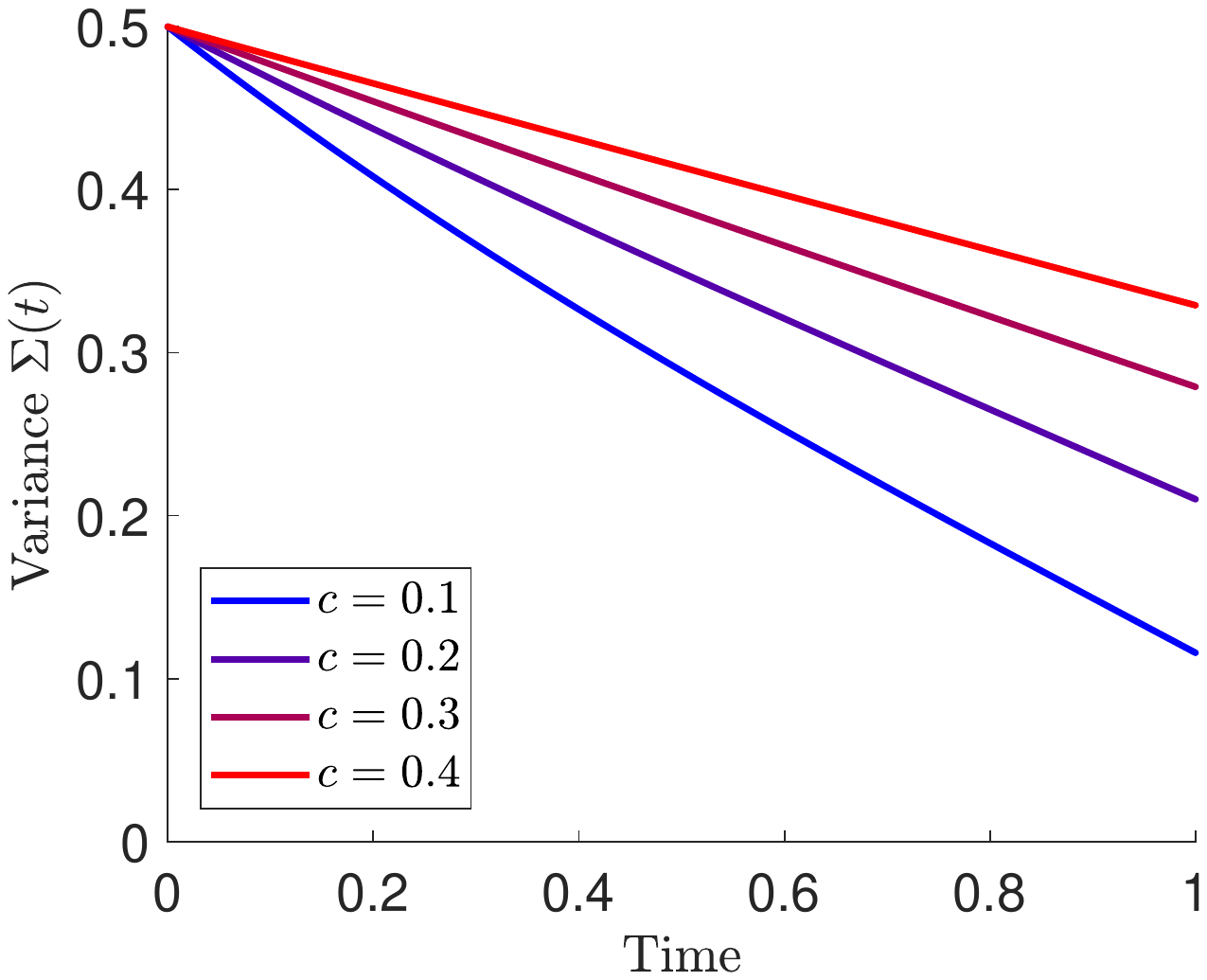}}
	\end{center}
	\vspace{-1em}
	\caption{Shown are the functions $\beta$, $\lambda$, and $\Sigma$ in equilibrium for various values of the transaction cost parameter, $c$. Other parameter values are $A = 1$, $\sigma = 1$, and $\Sigma^v_0 = 0.5$. \label{fig:dependence_c}}
\end{figure}

In Figure \ref{fig:dependence_A} we show the effect of varying the risk-aversion parameter $A$. These results are qualitatively similar to those found in \cite{baruch2002insider}. In particular, larger values of $A$ mean that the insider trades more aggressively at the beginning of the trading interval, as seen in the left panel by larger values of the trading rule $\beta$, because this is when they are exposed to the greatest amount of risk posed by the noise traders. At the end of the trading interval, the trading rule converges to the same finite value regardless of risk-aversion level because the remaining risk exposure converges to zero as $t\rightarrow T$. The finite limit is a consequence of the positive transaction cost $c$, whereas in \cite{baruch2002insider} this limit would be infinite. The more aggressive trading at the beginning of the trading interval means there is more information contained in the net order-flow, and thus we see in the middle panel that price impact is larger at earlier times compared to later times for any fixed value of risk-aversion $A>0$. The right panel demonstrates that price discovery also occurs faster when there is increased risk-aversion.

\begin{figure}
	\begin{center}
		{\includegraphics[trim=140 240 140 240, scale=0.45]{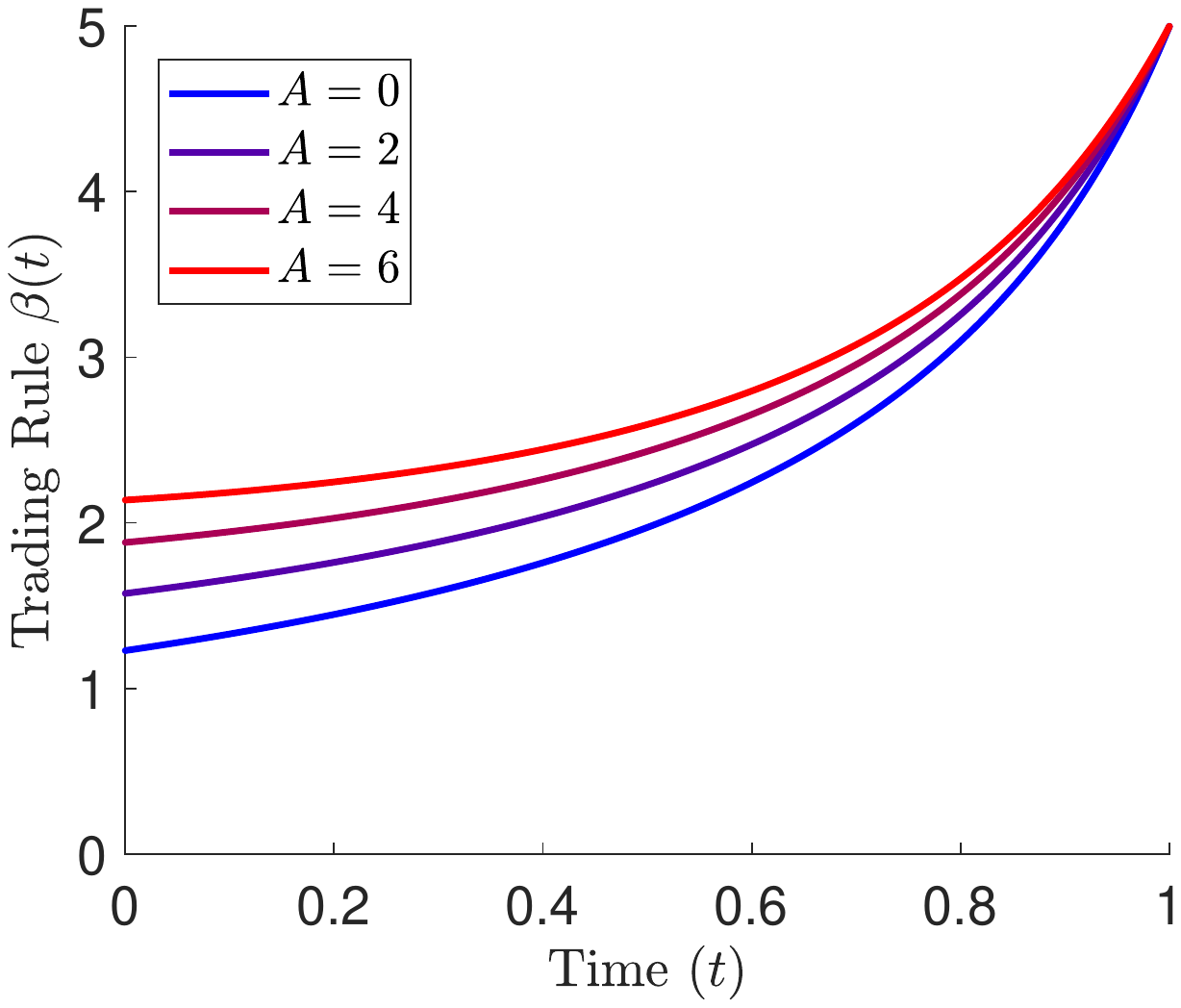}}\hspace{8mm}
		{\includegraphics[trim=140 240 140 240, scale=0.45]{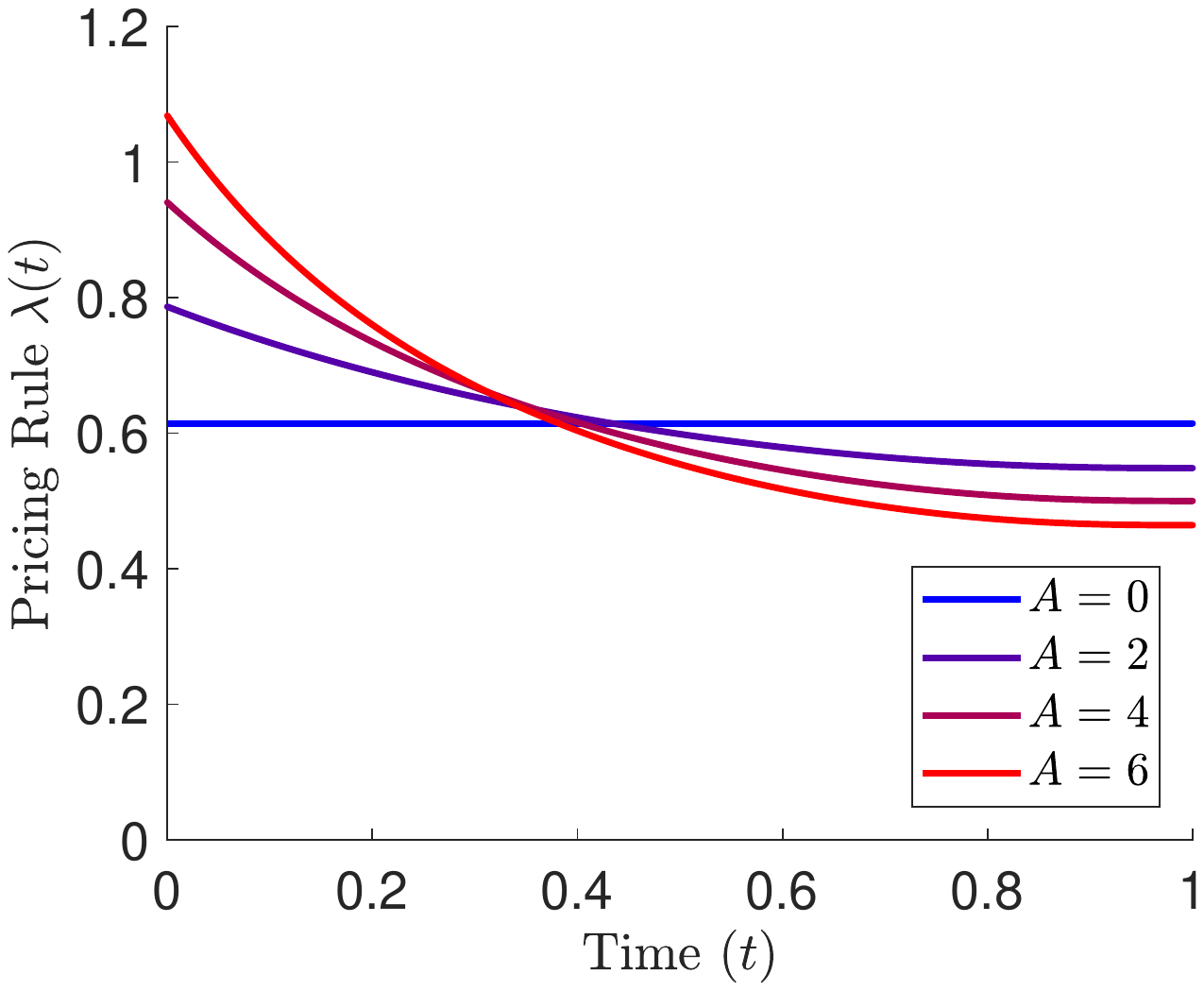}}\hspace{8mm}
		{\includegraphics[trim=140 240 140 240, scale=0.45]{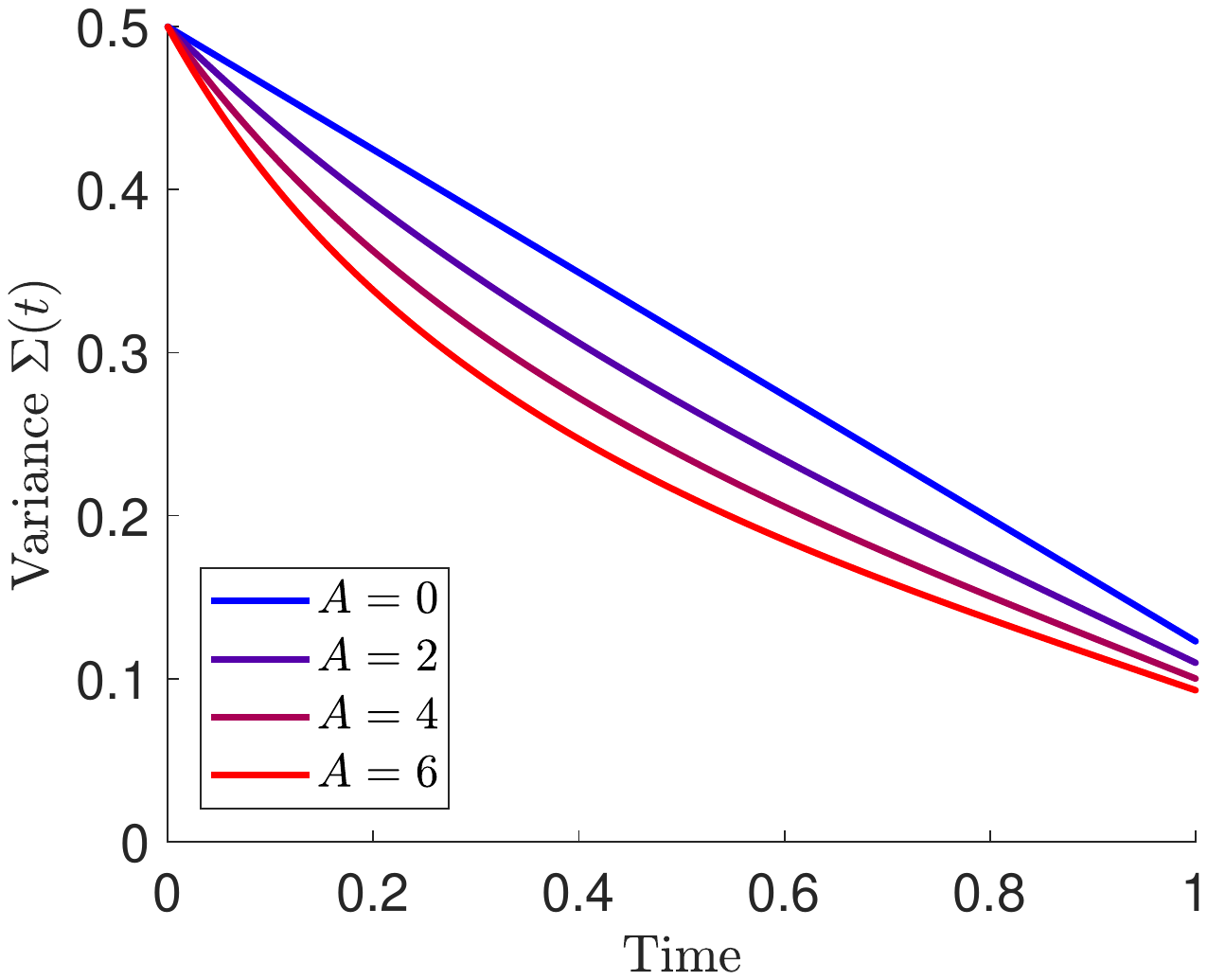}}
	\end{center}
	\vspace{-1em}
	\caption{Shown are the functions $\beta$, $\lambda$, and $\Sigma$ in equilibrium for various values of the risk-aversion parameter, $A$. Other parameter values are $c = 0.2$, $\sigma = 1$, and $\Sigma^v_0 = 0.5$. \label{fig:dependence_A}}
\end{figure}

In both Figure \ref{fig:dependence_c} and Figure \ref{fig:dependence_A} we note that the terminal conditional variance of the asset value is strictly positive. This is an immediate consequence of Lemma \ref{lem:genExistAndUnique} where it is shown that the solution to the FBODE \eqref{ode:FBODE} satisfies $x_1(t)>0$ for all $t\in[0,T]$. This means that the price process is not fully revealing of the asset's value as $t\rightarrow T$, unlike other frictionless models where equilibrium results in $\lim_{t\rightarrow T}P_t = v$ almost surely.

\subsection{Equilibrium Limits}
\label{sec:limitingEquilibria}

In this section, we study the limiting behaviour of the equilibrium rules $\beta$ and $\lambda$ of Theorem \ref{thm:continuousTimeEquilibrium} with respect to the transaction cost parameter $c$. The following proposition summarizes this behaviour as the transaction cost is taken towards its extremes.

\begin{proposition}[Limiting Transaction Cost Dependence]\label{prop:limits} For fixed $A\geq 0$,
	\begin{align}
		\lim_{c\rightarrow \infty} \beta(t; c) &= 0\,,\label{eqn:betacinf}\\
		\lim_{c\rightarrow \infty} \lambda(t; c) &= 0\,,\label{eqn:lambdacinf}\\
		\lim_{c\rightarrow 0} \beta(t; c)   
			&= 
			\left(
				\frac{A\Sigma_0^v}{2} +
				\sqrt{
					\left(
						\frac{A \Sigma_0^v}{2}
					\right)^2
					+ \Lambda_K^2
				}
			\right)
			\frac{
				1
			}{
				\Lambda_K^2 (T-t)
			}
			\,,\label{eq:betac0}\\
		\lim_{c\rightarrow 0} \lambda(t; c) 
			&= \frac{
				\Lambda_K^2
			}{
				\frac{A \Sigma_0^v}{2}\cdot \frac{2t - T}{T}
				+ \sqrt{
					\left(
						\frac{A \Sigma_0^v}{2}
					\right)^2
					+ \Lambda_K^2
				}
			}
				 \,.\label{eq:lambdac0}
	\end{align}
	The limits \eqref{eqn:betacinf} and \eqref{eqn:lambdacinf} hold uniformly for $t\in[0,T]$, and the limits \eqref{eq:betac0} and \eqref{eq:lambdac0} hold uniformly on any compact subinterval of $[0,T)$.
\end{proposition}
\begin{proof}
	For a proof see Section \ref{proof:limits} in the appendix.
\end{proof}

The two limits in \eqref{eqn:betacinf} and \eqref{eqn:lambdacinf} indicate that for sufficiently large transaction cost, the equilibrium essentially consists of the insider doing nothing and therefore there being no informational content to order flow (as it would be comprised of only noise trades). This result is expected given that any profits earned by the insider would be more than canceled out by the losses suffered from significant costs.

The two limits in \eqref{eq:betac0} and \eqref{eq:lambdac0} correspond to the equilibrium trading rule and pricing rule of in a setting where there is no transaction cost as in \cite{kyle1985continuous}, \cite{back1992insider}, and \cite{baruch2002insider}. This is significant because in our model the existence of a feedback form for the trading strategy relies on strict positivity of the transaction cost parameter. Typically in this style of model of asymmetric information, the HJB approach does not result in a feedback form for the optimal trading strategy. Proposition \ref{prop:limits} shows that in this case the correct equilibrium trading and pricing rules for a frictionless model can be obtained by adding an appropriate friction and then taking a limit as this friction vanishes. With a feedback form of the insider's trading strategy, some of the analysis becomes more straightforward, so this limiting technique could be applied in other settings of asymmetric information where the equilibrium rules are not as straightforward to classify.

An important step in obtaining this result is in identifying appropriate processes which will not have a discontinuity at $c=0$. When considering a vanishing friction term, some quantities can become discontinuous at $t=T$. In our model, this will happen with the function $h$ which appears in the insider's value function \eqref{eq:valueFunction}. To illustrate, in our model with $c>0$ we must have $h(T)=0$. But in frictionless equilibrium models it is typically the case that $\lim_{t\rightarrow T}h(t) > 0$ (in \cite{kyle1985continuous} $h$ is a positive constant). Thus it is impossible to have a well behaved limit as $c\rightarrow 0$. However, our model also results in $\Sigma(T)>0$ when $c>0$, which means that the product $\Sigma(t)\,h(t)>0$ for all $t\in[0,T)$ and $\Sigma(T)\,h(T)=0$. But these relations also hold in frictionless models, and so one of the appropriate functions to use in the classification of equilibrium is the product $\Sigma(t)\,h(t)$, which is represented by $x_2(t)$ in Lemma \ref{lem:genExistAndUnique} and Theorem \ref{thm:continuousTimeEquilibrium}.

\section{Conclusion}
\label{sec:conclusion}

In this paper, we extend the results of \cite{kyle1985continuous}, \cite{holden_subrahmanyam_1994} and \cite{baruch2002insider} by adding friction to the market in the form of a transaction cost which is linear in the insider's order size. We consider an exponentially risk-averse insider, allowing for the possibility of risk-neutrality. We begin by modeling a single-auction exchange and providing the corresponding unique market equilibrium in the form of an algebraic equation. We then demonstrate the effect of both the transaction cost and risk-aversion on the equilibrium by asymptotically expanding the algebraic equation about the frictionless risk-neutral equilibrium for small transaction cost and small risk-aversion. This procedure results in a single dimensionless quantity which determines an accurate expansion of the trading strategy to second order and the pricing rule to third order.

We then formulate an analogous market model in continuous-time. The mathematical machinery needed for the presentation of a linear market equilibrium in continuous-time is more sophisticated than that of the single-auction equilibrium. We develop an optimal trading strategy for an exponentially risk-averse insider by explicitly solving the HJB equation associated with optimization problem \eqref{eq:generalOptimalUtility}. We then give the filtering equations which provide an efficient price process for a fixed insider strategy. The main result is Theorem \ref{thm:continuousTimeEquilibrium}, which gives a linear, continuous-time market equilibrium in terms of the solution to the FBODE \eqref{ode:FBODE}, for which we show there exists a unique solution.

There are several qualitative similarities between the equilibria of the frictionless case and ours. First, linear equilibria exist in both cases and are given in terms of the solutions to differential equations. Additionally, the equilibrium insider trading rule $\beta$ is increasing in time, and the equilibrium market maker's pricing rule $\lambda$ is constant in the risk-neutral setting and decreasing in time in the risk-averse setting. However our model also has some properties which are not typical in many models of asymmetric information. The market maker's conditional variance $\Sigma$ has the property $\Sigma(T) > 0$, as distinct from the frictionless case in which $\Sigma(T) = 0$ meaning that the true value of the asset $v$ is not reveled to the market maker by time $T$.

The effects of varying model parameters is investigated by numerically computing the equilibrium strategies and by also explicitly computing some limits of equilibrium processes $\beta$ and $\lambda$ as transaction cost parameter $c$ tends to the values zero and infinity. We show that $\beta$ and $\lambda$ converge to the frictionless strategies of \cite{baruch2002insider} as $c \to 0$ and converge to zero as $c \to \infty$. Since the addition of friction to the continuous time model provides a feedback control, and the equilibrium processes $\beta$ and $\lambda$ converge to their frictionless counterparts when $c\rightarrow 0$, this could be used as a technique to study other frictionless models of asymmetry in which a feedback control is not directly obtained from the HJB approach.

\section{Appendix A}

\subsection{Proof of Theorem \ref{thm:singleAuctionEquil}} \label{sec:pf_thm:singleAucionEquil}

%\begin{proof}
	First we show that $r$ has a unique positive root. Define the function
	\begin{align}
	s(\lambda) &= 2\, ( \lambda + c )  + A\, \sigma^2\, \lambda^2\,. \label{eq:secondOrderCondition}
	\end{align}
	We will see later that $s(\lambda) > 0$ is the second order condition required for optimality of the trading strategy. The polynomial $r$ can be written as
	\begin{align}
	r(\lambda) &= \lambda\,(s^2(\lambda) + 4\, \lambda_K^2)- 4\, \lambda_K^2\, s(\lambda)\,, \label{eq:rInTermsOfs}
	\end{align}
	and we note that $r(\lambda) = 0$ if and only if 
	\begin{align}
	\lambda\,(s^2(\lambda) + 4\, \lambda_K^2 ) &= 4 \,\lambda_K^2\, s(\lambda)\,. \label{eq:alternativeLambdaEquation}
	\end{align}
	The left and right sides of \eqref{eq:alternativeLambdaEquation} are both polynomials in $\lambda$ with positive coefficients, and therefore are both strictly increasing functions of $\lambda \geq 0$. Furthermore, for any $A \geq 0$ the degree of the left hand side polynomial is strictly greater than the degree of the right hand side polynomial, and when $\lambda = 0$ the left hand side of \eqref{eq:alternativeLambdaEquation} is zero and the right hand side is positive. Therefore, \eqref{eq:alternativeLambdaEquation} has exactly one positive solution, and hence $r$ has exactly one positive root.
	
	The rest of the proof is divided into the two cases $A = 0$ and $A > 0$. Some related expressions are different between the two cases, but the structure of the proof in both cases is identical. First, we suppose that the insider's trading strategy is a linear function of $v$ and show the efficient pricing rule is linear. Then, we suppose the market maker's pricing rule is linear and show that the insider's optimal trading strategy is linear. Matching coefficients gives the result.
	\\
	
	\noindent
	\underline{Case $A = 0$}:
	First, we suppose that the market maker chooses $p$ as a linear function of $\Delta y$. Namely, we let $P(\Delta y) = \mu + \lambda\, \Delta y$, where $\mu$ and $\lambda$ are constants. When the market maker follows the pricing rule $P$, the insider's transaction price is given by
	\begin{align}
	\widehat p &= \mu + \lambda\, \Delta y + c\, \Delta x = \mu + (\lambda + c)\, \Delta x + \lambda \,\Delta z\,. \label{eq:linearTransactionPrice}
	\end{align}
	Thus, 
	\begin{align}
	\mathbb{E}[U(w)\,|\,v] = \mathbb{E}[(v - \widehat p)\,\Delta x|\,v]	&= \mathbb{E}[(v - \mu)\,\Delta x - (\lambda + c )\,\Delta x^2 - \lambda\, \Delta x\, \Delta z\,|\, v] \\
	&= (v - \mu )\,\Delta x - (\lambda + c )\,\Delta x^2\,.	\label{eq:ewProof}
	\end{align}
	The second order condition for optimality is $\lambda+c>0$, which is equivalent to $s(\lambda)>0$. Assuming this is satisfied, the choice of
	\begin{align}
	\Delta x &\equiv X(v) =	\frac{v - \mu}{2\, (\lambda + c)}\,, \label{eq:deltaxProofA0}
	\end{align}
	maximizes \eqref{eq:ewProof} for any $v$. We note that for linear pricing rules $P$ linear trading strategies $X$ are optimal even if we allow $X$ to be a nonlinear function.
	
	Now we assume that the insider chooses $\Delta x$ to be a linear function $X$ of the \textit{ex-post} price $v$. Specifically, the insider chooses $X (v) = \alpha + \beta\, v$, where $\alpha$ and $\beta$ are constants. Then, by the projection theorem for normal random variables we have
	\begin{align}
	\mathbb{E}[ v \,|\, \Delta y]	= \mathbb{E}[ v \,|\, \alpha + \beta\, v + \Delta z ] &= \mathbb{E}[v] + \frac{ \mathbb{E}[ (v - \mathbb{E}[v] )(\Delta y - \mathbb{E}[\Delta y] ) ] }{ \mathbb{E}[( \Delta y - \mathbb{E}[\Delta y]  )^2 ] } ( \Delta y - \mathbb{E}[\Delta y] ) \nonumber\\
	&= v_0 + \frac{	4\, \lambda_K^2\, \beta }{ 1 + 4\, \lambda_K^2\, \beta^2 } (\Delta y - \alpha - \beta\, v_0 ) \nonumber\\
	&= \frac{ v_0 - 4 \,\lambda_K^2\, \alpha\, \beta }{ 1 + 4 \,\lambda_K^2 \,\beta^2 } + \frac{ 4\, \lambda_K^2\, \beta }{ 1 + 4\, \lambda_K^2\, \beta^2 } \,\Delta y\,.\label{eq:vCondExpectationDiscTime}
	\end{align}
	
	By examining \eqref{eq:deltaxProofA0} and \eqref{eq:vCondExpectationDiscTime}, we see that for $(P,X) \equiv (P(\Delta y), X(v)) = ( \mu + \lambda\, \Delta y, \alpha + \beta \,v)$ to be an equilibrium, it must be that
	\begin{subequations}
		\begin{align}
		\alpha 	&= - \beta \,\mu, & \mu &= \frac{ v_0 - 4\, \lambda_K^2\, \alpha \,\beta }{ 1 + 4\, \lambda_K^2 \,\beta^2 }\,, \label{eqs:alphaMuProofA0}\\
		\beta &= \frac{ 1 }{ 2\, (\lambda + c) }, & \lambda &= \frac{ 4\, \lambda_K^2 \,\beta }{ 1 + 4 \,\lambda_K^2 \,\beta^2 }\,, \label{eqs:betaLambdaProofA0}
		\end{align}
	\end{subequations}
	subject to the constraint $s(\lambda) > 0$. 
	
	Immediately, we get from \eqref{eqs:alphaMuProofA0} that in equilibrium $\mu = v_0$ and $\alpha = - \beta\, v_0$, which gives the insider's trading strategy $X(v) = \beta\, ( v - v_0)$. Recall that the insider's second order condition is satisfied if $s(\lambda) > 0$. Inserting $\beta$ into the expression for $\lambda$ in \eqref{eqs:betaLambdaProofA0}, we see that $\lambda$ satisfies \eqref{eq:alternativeLambdaEquation}. If $\lambda \leq 0$ then from \eqref{eq:alternativeLambdaEquation} we see that $s(\lambda)\leq 0$, contradicting optimality of $\Delta x$. Therefore, $\lambda$ is the unique positive root of the polynomial $r$.
	\\
	
	\noindent
	\underline{Case $A > 0$}:
	First, suppose that the market maker chooses $p$ as a linear function of $\Delta y$ so that $P(\Delta y) = \mu + \lambda \Delta y$, where $\mu$ and $\lambda$ are constants. When the market maker follows the pricing strategy $P$, then the insider's transaction price $\widehat p$ is given by \eqref{eq:linearTransactionPrice}. Thus, 
	\begin{align}
	\mathbb{E}[U(w)| v] &= \mathbb{E}[ -\exp\{- A ( v- \widehat p ) \Delta x \}| v] \nonumber\\
	&= - \exp \{ -A (v - \mu )\Delta x + A (\lambda + c )\Delta x^2 + \tfrac 12 A^2 \lambda^2 \sigma^2 \Delta x^2\}\,. \label{eq:euwProof} 
	\end{align}
	The second order condition for optimality is $2(\lambda+c)+A\sigma^2\lambda^2 > 0$ which is equivalent to $s(\lambda) > 0$. Assuming this is satisfied, the choice of 
	\begin{align}
	\Delta x & \equiv X(v) = \frac{v - \mu}{ 2\, (\lambda + c) + A \,\sigma^2\, \lambda^2} \label{eq:deltaxProofAgt0}
	\end{align}
	maximizes \eqref{eq:euwProof} for any $v$. We note that for linear pricing rules $P$ linear trading strategies $X$ are optimal even if we allow $X$ to be a nonlinear function.
	
	Now, suppose that the insider chooses the linear trading strategy $X(v) = \alpha + \beta\, v$. Then $\mathbb{E}[v \,|\,\Delta y]$ is given by \eqref{eq:vCondExpectationDiscTime}.
	
	By examining \eqref{eq:vCondExpectationDiscTime} and \eqref{eq:deltaxProofAgt0}, we see that for $(P,X) \equiv (P(\Delta y), X(v)) = ( \mu + \lambda \,\Delta y, \alpha + \beta\, v)$ to be an equilibrium we must have
	\begin{subequations}
		\begin{align}
		\alpha &= - \beta \,\mu, & \mu &= \frac{ v_0 - 4 \,\lambda_K^2 \,\alpha \,\beta }{ 1 + 4 \,\lambda_K^2\, \beta^2 }\,, \label{eqs:alphaMuProofA1}\\
		\beta &= \frac{	1 }{ 2\, (\lambda + c) + A \,\sigma^2 \,\lambda^2 }\,, & \lambda &= \frac{ 4 \,\lambda_K^2 \,\beta }{ 1 + 4 \,\lambda_K^2 \,\beta^2 }\,, \label{eqs:betaLambdaProofA1}
		\end{align}
	\end{subequations}
	where $s(\lambda) > 0$. The rest of the proof is identical to the case $A = 0$. \qed
%\end{proof}

\subsection{Proof of Proposition \ref{prop:single_auction_dependence}} \label{sec:pf_prop:single_auction_dependence}

%\begin{proof}
	In the proof of Theorem \ref{thm:singleAuctionEquil}, we saw that when $A > 0$ the equilibrium $\beta$ and $\lambda$ satisfy \eqref{eqs:betaLambdaProofA1}. Inserting $\lambda$ into $\beta$ in \eqref{eqs:betaLambdaProofA1} we get that in equilibrium $\beta$ satisfies
	\begin{align}
	q(\beta) &:= 32\,c\, \lambda_K^4\,\beta^5 + 16\, \lambda_K^4\, \beta^4 + 16 \, \lambda_K^2 \,(A \,\lambda_K^2 \,\sigma ^2 + c)\,\beta ^3 + 2 \,c\, \beta - 1 = 0\,,\label{eq:betaPoly}
	\end{align}
	with $\beta >0$. We see that for fixed $x > 0$ we have $\partial q (x) / \partial c > 0$ and $\partial q (x) / \partial A > 0$. Since $q(x)$ is a strictly increasing function of $x$, we therefore must have $\partial \beta / \partial c < 0$ and $\partial \beta / \partial A < 0$. Taking derivatives of $\lambda$ in \eqref{eqs:betaLambdaProofA1} with respect to $c$ and $A$ yields 
	\begin{align}
		\frac{\partial \lambda}{\partial c} &= \frac{4\,\lambda_K^2\,(1-4\,\lambda_K^2\,\beta^2)}{(1+4\,\lambda_K^2\,\beta^2)^2} \, \frac{\partial\beta}{\partial c}\,,\\
		\frac{\partial \lambda}{\partial A} &= \frac{4\,\lambda_K^2\,(1-4\,\lambda_K^2\,\beta^2)}{(1+4\,\lambda_K^2\,\beta^2)^2} \, \frac{\partial\beta}{\partial A}\,.
	\end{align}
	Therefore,
	\begin{align}
		\frac{\partial \lambda}{\partial c} < 0\,, \frac{\partial \lambda}{\partial A} < 0 \quad \Longleftrightarrow \quad \beta < \frac{ 1 }{2\, \lambda_K} = \beta_K\,.
	\end{align}
	We have already shown that $\beta$ is strictly decreasing with respect to $c$ and $A$, and in particular we have $\beta \leq \beta_K$ with equality only when $c=A=0$. Thus, $\lambda$ is strictly decreasing with respect to both $c$ and $A$. \qed
%\end{proof}

\subsection{Proof of Proposition \ref{prop:single_auction_approx}} \label{sec:pf_prop:single_auction_approx}

%\begin{proof}
	The proof of this proposition relies on the result that the roots of a polynomial depend on the coefficients analytically in a neighbourhood of a given root (see \cite{brillinger1966analyticity}). With this in mind we write $\lambda$ and $\beta$ as a power series in the quantities $A$ and $c$ as follows:
	\begin{align}
		\lambda &= \sum_{i=0}^\infty \sum_{j=0}^\infty \tilde{\lambda}_{i,j} \,A^i\,c^j\,,\\
		\beta   &= \sum_{i=0}^\infty \sum_{j=0}^\infty \tilde{\beta}_{i,j} \,A^i\,c^j\,.\label{eqn:pf_beta_expand}
	\end{align}
	Our approximation corresponds to computing each $\tilde{\lambda}_{i,j}$ for $i+j\leq 3$ and each $\tilde{\beta}_{i,j}$ for $i+j\leq 2$, the remainder of the higher order terms being either $o(A^4+c^4)$ (for $\lambda$) or $o(A^3+c^3)$ (for $\beta$).
	
	We first substitute the expansion for $\lambda$ into the polynomial $r(x)$ given in \eqref{eq:singleAuctionLambdaPoly} and recall that $\lambda$ is the unique positive root of this polynomial. We then collect terms according to the powers of $A$ and $c$ and set each term equal to zero individually. This results in a system of $10$ equations\footnote{The equations are large and tedious, so are displayed in a subsequent appendix} which does not have a unique solution. However, inspection shows that the coupling of the equations is arranged in such a way that the first appearance of each quantity to be solved for is linear, with the exception of $\tilde{\lambda}_{0,0}$. Thus, given $\tilde{\lambda}_{0,0}$, every other $\tilde{\lambda}_{i,j}$ is solved for uniquely. Because we require $\lambda$ to be positive, we must use the solution which corresponds to a positive value of $\tilde{\lambda}_{0,0}$, of which there is only one (the other roots are zero and negative). The resulting solution is
	\begin{align*}
	\tilde{\lambda}_{00} &= \lambda_K\,, &	\tilde{\lambda}_{10} &= 0\,, & \tilde{\lambda}_{01} &= 0\,, &	\tilde{\lambda}_{20} &= -\frac{1}{8}\,\sigma^4\,\lambda_K^3\,, & \tilde{\lambda}_{11} &= -\frac{1}{2}\,\sigma^2\,\lambda_K\,, \\
	 \tilde{\lambda}_{02} &= -\frac{1}{2}\,\frac{1}{\lambda_K}\,, & \tilde{\lambda}_{30} &= \frac{1}{8}\,\sigma^6\,\lambda_K^4\,, & \tilde{\lambda}_{21} &=  \frac{3}{4}\,\sigma^4\,\lambda_K^2\,, & \tilde{\lambda}_{12} &= \frac{3}{2}\,\sigma^2 \,, & \tilde{\lambda}_{03} &= \frac{1}{\lambda_K^2}\,.\\
	\end{align*}
	By letting $\displaystyle\nu = \frac{1}{2}\,\lambda_K\,\sigma^2\,A + \frac{c}{\lambda_K}$ we are able to write the truncated sum corresponding to the above coefficients and perform some elementary factoring which yields
	\begin{align*}
		\sum_{i+j\leq 3} \tilde{\lambda}_{i,j}\, A^i\, c^j &= \lambda_K\,\biggl(1 - \frac{1}{2}\,\nu^2 + \nu^3\biggr)\,.
	\end{align*}
	This is the desired form as given in the statement of the Proposition.
	
	Solving for $\tilde{\beta}_{i,j}$ for $i+j\leq 2$ is performed through essentially the same process, except the expansion \eqref{eqn:pf_beta_expand} is substituted into the polynomial $q(x)$ defined in \eqref{eq:betaPoly}. By collecting powers of $A$ and $c$ and setting each to zero individually, this once again yields a system of $6$ equations\footnote{Similarly, shown in a subsequent appendix}. The resulting system of equations does not have a unique solution, but given $\tilde\beta_{0,0}$ the remainder of the equations to be solved are linear. There is only one value of $\tilde\beta_{0,0}$ which ensures $\beta$ is positive (the other roots being negative or imaginary). Solving the remaining linear equations yields
	\begin{align*}
	\tilde\beta_{00} &= \frac{1}{2\,\lambda_K}\,, &  \tilde\beta_{10} &= -\frac{1}{4}\,\sigma^2\,, & \tilde\beta_{01} &= -\frac{1}{2\,\lambda_K^2}\,,\\
	\tilde\beta_{20} &= \frac{3}{16}\,\sigma^4\,\lambda_K\,, & \tilde\beta_{11} &= \frac{3}{4}\,\frac{\sigma^2}{\lambda_K}\,, & \tilde\beta_{02} &= \frac{3}{4\,\lambda_K^3}\,.
	\end{align*}
	The truncated sum corresponding to these coefficients can again be factored to give
	\begin{align*}
		\sum_{i+j\leq 2} \tilde{\beta}_{i,j}\, A^i\, c^j &= \frac{1}{2\,\lambda_K}\,\biggl(1 - \nu + \frac{3}{2}\,\nu^2\biggr)\,.
	\end{align*}
	
	This is the desired form as given in the statement of the Proposition. \qed
%\end{proof}

\subsection{Proof of Lemma \ref{lem:optimalTradingStrategy}} \label{proof:optimalTradingStrategy}

%\begin{proof}
	We consider the cases $A = 0$ and $A > 0$ separately. \\
	
	\noindent
	\underline{Case $A = 0$}: When $A = 0$,  $U(w) = w$, and the value function \eqref{eq:generalValueFunction} becomes 
	\begin{align}	
		H(t,P) &= \sup_{\theta \in \mathcal A} \mathbb{E}\biggl[ \int_t^T (v -  P_s - c\, \theta_s )\, \theta_s \, ds	\, \biggl|\, 	\mathcal{F}_t^I	\biggr] \,.\label{eq:riskNeutralValueFunctionDefinition}
	\end{align} 
	Associated with this stochastic control problem is the HJB partial differential equation 
	\begin{align}
		\partial_t H + \sup_\theta \biggl\{ \frac {1}{2}\, \sigma^2\, \lambda^2(t)\, \partial_{PP} H 	+ \theta\, \lambda(t)\, \partial_{P} H +  ( v - P - c \,\theta)\, \theta \biggr\} &= 0\,, & H(T,\cdot) &= 0\,.\label{eq:riskNeutralHJB}
	\end{align}
	It can be checked by direct substitution that the solution of this equation is given by \eqref{eq:valueFunction} when $h$ satisfies the ODE \eqref{ode:riccatiODE}. The supremum in \eqref{eq:riskNeutralHJB} is achieved at
	\begin{align}
		\theta^*( t, P) &= \beta(t)\,( v - P), & \beta(t) &= \frac{ 1 - 2 \,\lambda(t)\, h(t) }{ 2\, c }\,.\label{eq:optimalThetaProof}
	\end{align}
	All that remains to be shown is that this feedback form of $\theta^*$ yields an admissible trading strategy. Optimality of $\theta^*$ then follows from a standard verification argument (see \cite{pham2009continuous}). To this end, define an auxilliary process $Q = (Q_t)_{0 \leq t \leq T}$ by
	\begin{align*}
		Q_t &= P_t - v\,.
	\end{align*}
	Then under the control $\theta^*$, we have the dynamics
	\begin{align*}
		dQ_t &= -\lambda(t)\,\beta(t)\,Q_t\,dt + \lambda(t)\,\sigma\,dZ_t\,, & Q_0 &= v_0 - v\,.
	\end{align*}
	This stochastic differential equation is linear and therefore has a unique strong solution. Further, since $v_0-v$ is Gaussian, the resulting solution is a Gaussian process (see \cite{karatzas2012brownian} Section 5.6). This gives
	\begin{align*}
		\mathbb{E}\biggl[\int_0^T (\theta_s^*)^2 \, ds\biggr] &= \mathbb{E}\biggl[\int_0^T \beta^2(s)Q_s^2 \, ds\biggr] < \infty\,.
	\end{align*}
	In addition, $Q$ has continuous paths and thus the trading strategy is predictable, and therefore admissible.

	\noindent
	\underline{Case $A > 0$}: When $A > 0$, $U(w) = - \exp(-Aw)$, and the value function \eqref{eq:generalValueFunction} becomes
	\begin{align}
		H(t, P)	&= \sup_{\theta \in \mathcal A} \mathbb{E}\biggl[ - \exp \left\{ - A \int_t^T (   v - P_s -  c \,\theta_s )\, \theta_s \, ds \right\} \,\biggl|\, \mathcal{F}_t^I \biggr]\,.\label{eq:riskAverseValueFunctionDefinition}
	\end{align}
	This stochastic control problem has the associated HJB partial differential equation
	\begin{align}
		\partial_t H + \sup_\theta \biggl\{\frac {1}{2}\, \sigma^2\, \lambda^2(t)\, \partial_{PP} H + \theta \,\lambda(t)\, \partial_{P} H - A \,( v - P - c\, \theta)\, \theta\, H	\biggr\} &= 0, & H(T,\cdot) &= -1\,.\label{eq:riskAverseHJB}
	\end{align}
	Once again, it can be checked by direct substitution that this equation has solution given by \eqref{eq:valueFunction}. The resulting feedback form of the control is again linear with respect to $v-P$, and thus the remainder of the proof is identical to the risk-neutral case. \qed
%\end{proof}

\subsection{Proof of Lemma \ref{lem:genExistAndUnique}} \label{proof:genExistAndUnique}
In this section, we present a proof of Lemma \ref{lem:genExistAndUnique} for the case $A>0$. The proof for $A=0$ is essentially the same but more straightforward, and the solution given by \eqref{eq:odeAlphaZeroSol} can be checked by direct substitution.

\begin{proof}[Proof of Lemma \ref{lem:genExistAndUnique}]
	We first consider $c>0$. Inspection of equation \eqref{ode:FBODE} shows that any solution must have $x_1$ a decreasing function, so we look for a solution in which $x_2(t) = \rho(x_1(t))$ for some function $\rho$. This gives
	\begin{align}
		\frac{d\rho}{d x_1} &= \frac{dx_2 /dt }{dx_1 /dt}\nonumber\\
		&= \frac{c\,\sigma^2}{x_1} + \frac{\rho}{x_1} - \frac{2\,A\,\rho^2}{x_1}\,.\label{eqn:rhoode}
	\end{align}
	This ODE has general solution
	\begin{align}
		\rho(x_1) &= \frac{(\gamma+1)\,x_1^\gamma - (\gamma-1)\,k}{4\,A\,(x_1^\gamma + k)}\,,\label{eqn:rhox1}\\
		\gamma &= \sqrt{1 + 8\,A\,c\,\sigma^2}\,,\label{eqn:gamma}
	\end{align}
	for arbitrary $k\in\mathbb{R}$. The value of $k$ must be chosen to match the boundary condition $x_2(T) = \rho(x_1(T)) = 0$. We can immediately rule out $k=0$ because then $x_2(t) = \rho(x_1(t))$ becomes a non-zero constant contradicting the boundary condition $x_2(T) = 0$. Therefore, we search for non-zero $k$ such that $g(k) = 0$, where
	\begin{align*}
		g(k) &= (\gamma+1)\,x_1^\gamma(T;k) - (\gamma-1)\,k\,,
	\end{align*}
	and where by substituting \eqref{eqn:rhox1} into \eqref{eq:F} we have that $x_1$ satisfies
	\begin{align}
		\frac{dx_1}{dt}(t;k) &= f_1(x_1(t;k);k)\,, & x_1(0;k) &= \Sigma_0^v\,,\label{eqn:ODEx1}\\
		f_1(x_1;k) &= -\frac{ 16\, A^2\, \sigma^2\, x_1^2\, \left(x_1^\gamma  + k \right)^2 }{ \biggl((\gamma +1)^2\, x_1^\gamma  + (\gamma -1)^2\, k\biggr)^2 }\,, 
	\end{align}
	where we have made the dependence of $x_1$ on the parameter $k$ explicit. Inspection of \eqref{eqn:ODEx1} shows that $x_1(t;k)$ must be positive for all $t\in[0,T]$, and so we must have $k>0$ in order to have $g(k) = 0$. We now show that there exists exactly one value of $k>0$ which provides this solution.

	To this end we show that for $k>0$, $x_1(T;k)$ is decreasing with respect to $k$. As $f_1$ is continuously differentiable with respect to $k$, we have that $x_1(T;k)$ is as well (see \cite{hartman_ordinary_2002}). By setting $z(t;k) = \partial x_1(t;k)/\partial k$, we have
	\begin{align}
		\label{ode:z}
		\frac{\partial z(t;k)}{\partial t} 
			&= \frac{
				\partial f_1 ( x_1(t;k);k)
			}{
				\partial x_1
			} z(t;k) 
			+ \frac{
				\partial f_1 (x_1(t;k);k)
			}{
				\partial k
			}, & 
		z(0;k)
			&= 0,
	\end{align}
	where
	\begin{align}
		\frac{
				\partial f_1 ( x_1; k)
			}{
				\partial x_1
			}
			&=
			-\frac{
				32\, A^2\, \sigma ^2\, x_1\, 
				\left(k+x_1^{\gamma }\right)\, 
				\biggl((\gamma +1)\, x_1^{\gamma } - (\gamma -1 )\,k \biggr)^2
			}{
				\biggl((\gamma +1)^2\, x_1^{\gamma } + (\gamma -1)^2\, k\biggr)^3
			}\, , \\
		\label{eq:df1dk}
		\frac{
				\partial f_1 ( x_1; k)
			}{
				\partial k
			}	
			&= 
			-\frac{
				128\, A^2\, \gamma\,  \sigma ^2\, x_1^{\gamma +2}
				(x_1^{\gamma } + k )
			}{
				\biggl((\gamma +1)^2\, x_1^{\gamma } + (\gamma -1)^2\, k\biggr)^3
			}\,.
	\end{align}
	Since $x_1(t;k)>0$ for $t\in[0,T]$ and $k>0$, we have
	\begin{align}
	\frac{
			\partial f_1 ( x_1(t; k); ,k)
		}{
			\partial x_1
		}
		& \leq 0, & 
	\frac{
			\partial f_1 ( x_1(t; k); k)
		}{
			\partial k
		}
		& < 0\,.
	\end{align} 
	Thus, by examining \eqref{ode:z}, we see that $ \partial x_1(t;k)/\partial k < 0$ for all $t\in (0, T]$ and $k>0$. This establishes that $x(T;k)$ is decreasing with respect to $k$, and therefore any solution to $g(k)=0$ must be unique. Let $k_r$ and $k_l$ be given by
	\begin{align*}
		k_r &= \frac{\gamma+1}{\gamma-1}\,(\Sigma_0^v)^\gamma\,, & k_l &= \frac{\gamma+1}{\gamma-1}x^\gamma_1(T;k_r)\,.
	\end{align*}
	Then we see $0<k_l<k_r$ and $g(k_l)>0>g(k_r)$. Thus we have a unique $k>0$ such that $x_2(T;k) = 0$. This establishes existence and uniqueness of the solution to \eqref{ode:FBODE} along with the property that $x_1(t)>0$ for all $t\in[0,T]$. The fact that $x_2(t)>0$ for all $t\in[0,T)$ then follows immediately from \eqref{eqn:rhox1} together with the boundary condition $x_2(T)=0$ and that $x_1(t)$ is decreasing.

	We now consider $c = 0$ and proceed similarly by looking for a function $\rho$ such that $x_2(t) = \rho(x_1(t))$. In this case the ODE satisfied by $\rho$ is
	\begin{align*}
		\frac{d\rho}{d x_1} &= \frac{\rho}{x_1} - \frac{2\,A\,\rho^2}{x_1}\,,
	\end{align*}
	which has solution
	\begin{align*}
		\rho(x_1) &= \frac{x_1}{2\,A\,(x_1+k)}\,,
	\end{align*}
	for arbitrary $k\in\mathbb{R}$. We immediately eliminate $k=0$ for the same reason as before. However, we cannot eliminate $k<0$ from consideration. The ODE for $x_1$ now takes the form
	\begin{align*}
		\frac{dx_1}{dt}(t;k) &= -A^2\,\sigma^2\,(x_1(t;k) + k)^2\,, & x(0;k) &= \Sigma_0^v\,,
	\end{align*}
	which has the unique solution
	\begin{align}
		x_1(t;k) &= \frac{\Sigma_0^v - k(\Sigma_0^v+k)A^2\sigma^2t}{1 + (\Sigma_0^v+k)A^2\sigma^2t}\,.\label{eqn:x1c0}
	\end{align}
	Enforcing the boundary condition $x_2(T;k) = 0$ yields two possible values of $k$:
	\begin{align}
		k &= -\frac{\Sigma_0^v}{2}\biggl(1 \pm \sqrt{1 + \frac{4}{A^2\,\sigma^2\, T\,\Sigma_0^2}}\biggr)\,.\label{eqn:kpm}
	\end{align}
	This results in one positive and one negative value of $k$, both of which provide solutions to the FBODE \eqref{ode:FBODE}. However, the negative value of $k$ results in $x_2(t)<0$ for all $t\in[0,T)$. By enforcing $x_2(t)>0$ for $t\in[0,T)$ we must discard the negative value of $k$ and are left with a unique solution. From \eqref{eqn:x1c0} we have immediately that $x_1(t;,k)>0$ for all $t\in[0,T)$ and $x_1(T) = 0$.
	\qed
\end{proof}

\subsection{Proof of Theorem \ref{thm:continuousTimeEquilibrium}} \label{proof:continuousTimeEquilibrium}

%\begin{proof}
	First we fix $(x_1, x_2)$ to be the unique solution to the FBODE \eqref{ode:FBODE} with $x_1(0) = \Sigma_0^v$. Suppose the price dynamics are given by
	\begin{align*}
		dP_t &= \lambda(t)\,dY_t\,, & P_t &= v_0\,,
	\end{align*}
	with $\lambda(t)$ given by
	\begin{align*}
		\lambda(t) &= \frac{\Sigma(t)}{2\,(c\,\sigma^2 + \Sigma(t)\,h(t))}\,,
	\end{align*}
	where $\Sigma(t) = x_1(t)$ and $h(t) = x_2(t)/x_1(t)$ as in the statement of the Theorem. A straightforward computation shows that
	\begin{align*}
		\frac{dh(t)}{dt} &= \frac{2\,A\,\sigma^2\,\Sigma^2(t)\,h^2(t) -c\,\sigma^4}{4\,(c\,\sigma^2 + \Sigma(t)\,h(t))^2}\,,
	\end{align*}
	and in addition, that
	\begin{align*}
		-\frac{(1-2\,A\,c\,\sigma^2)\,\lambda^2(t)}{c}\,h^2(t) + \frac{\lambda(t)}{c}\,h(t) - \frac{1}{4\,c} &= \frac{2\,A\,\sigma^2\,\Sigma^2(t)\,h^2(t) - c\,\sigma^4}{4\,(c\,\sigma^2 + \Sigma(t)\,h(t))^2}\,.
	\end{align*}
	Thus we have that $h$ satisfies the ODE \eqref{ode:riccatiODE}. Therefore, by Lemma \ref{lem:optimalTradingStrategy}, the optimal trading strategy is given by
	\begin{align*}
		\theta^*_t &= \beta(t)\,(v-P_t)\,dt\,,
	\end{align*}
	with
	\begin{align*}
		\beta(t) &= \frac{ 1 - 2\, \lambda(t)\,  h(t) }{ 2\, c }\,,
	\end{align*}
	and the insider's value function is given by \eqref{eq:valueFunction}.
	
	Now, again with $(x_1, x_2)$ the unique solution to the FBODE, suppose that the insider's trading strategy is given by
	\begin{align}
		\theta_t &= \beta(t)\,(v-P_t)\,dt\,,\label{eqn:theta_pf}
	\end{align}
	with
	\begin{align*}
		\beta(t) &= \frac{\sigma^2}{2\,(c\,\sigma^2 + \Sigma(t)\,h(t))}\,,
	\end{align*}
	where $\Sigma(t) = x_1(t)$ and $h(t) = x_2(t)/x_1(t)$. From \eqref{ode:FBODE} we see that $\Sigma$ satisfies
	\begin{align}
		\frac{d\Sigma(t)}{dt} &= -\frac{\sigma^2\,\Sigma^2(t)}{4\,(c\,\sigma^2 + \Sigma(t)\,h(t))^2} = -\sigma^2\,\lambda^2(t)\,.\label{eqn:dSigma_pf}
	\end{align}
	In addition we have by their definitions in the statement of the Theorem that
	\begin{align}
		\lambda(t) &= \frac{\beta(t)\,\Sigma(t)}{\sigma^2}\,.\label{eqn:lambda_pf}
	\end{align}
	By Lemma \ref{lem:filtering}, the relations in \eqref{eqn:theta_pf}, \eqref{eqn:dSigma_pf}, and \eqref{eqn:lambda_pf} imply that the efficient pricing rule is given by
	\begin{align*}
		dY_t &= \lambda(t)\,dY_t\,, & P_0 &= v_0\,,
	\end{align*}
	and that the function $\Sigma$ also yields the market maker's conditional variance:
	\begin{align*}
		\Sigma(t) &= \mathbb{E}[(v-P_t)^2|\mathcal{F}_t^M]\,.
	\end{align*}
	
	Finally, when $A = 0$, applying Lemma \ref{lem:genExistAndUnique} gives the closed form expressions
	\begin{align*}
		\beta(t) &= \frac{ 	1 }{ \lambda\, (T-t) + 2\, c	}\,, \\
		\lambda &= \sqrt{ \Lambda_K^2 + \frac{c^2}{T^2}} - \frac {c}{ T }\,.
	\end{align*}
	as desired.	\qed
%\end{proof}

\subsection{Proof of Proposition \ref{prop:compareStat}} \label{proof:compareStat}

As in the proof of Theorem \ref{thm:continuousTimeEquilibrium}, let $(x_1,x_2)$ be the unique solution to the FBODE \eqref{ode:FBODE}, and let $\Sigma_(t) = x_1(t)$ and $h(t) = x_2(t)/x_1(t)$. Then writing
\begin{align*}
	\beta(t) &= \frac{ \sigma^2 }{ 2 \,( c\, \sigma^2 + \Sigma(t)\,h(t) ) }\,, & \lambda(t) &= \frac{ \Sigma(t) }{ 2 \,( c\, \sigma^2 + \Sigma(t)\,h(t) ) }\,,
\end{align*}
we have that $\beta$ is increasing if and only if $x_2$ is decreasing. Recalling that
\begin{align*}
	\frac{x_2(t)}{dt} &= -\frac{\sigma^2\,x_1\,(c\,\sigma^2 + x_2 - 2\,A\,x_2^2)}{4\,(c\,\sigma^2 + x_2)^2}\,,
\end{align*}
we have that if $x_2(t) > 1 + \sqrt{1 + 8\,A\,c\,\sigma^2}/(4\,A)$, the $x_2'(t) > 0$. However, satisfying these inequalities for any $t$ would violate the boundary condition $x_2(T) = 0$, so we must have that $x_2$ is decreasing and $\beta$ is increasing.

A straightforward computation shows that
\begin{align*}
	\frac{d\lambda(t)}{dt} &= -\frac{A\,\sigma^2\,\Sigma^4(t)\,h^2(t)}{4\,(c\,\sigma^2 + \Sigma(t)\,h(t))^4}\,,
\end{align*}
and so if $A>0$ we have that $\lambda$ is decreasing. \qed

\subsection{Proof of Proposition \ref{prop:limits}}\label{proof:limits}

From Theorem \ref{thm:continuousTimeEquilibrium} and Proposition \ref{prop:compareStat}, the function $\beta$ is positive, increasing, and satisfies $\beta(T) = 1/2\,c$, therefore $\beta$ is uniformly bounded by $1/2\,c$. The uniform limit
\begin{align}
	\lim_{c\rightarrow\infty}\beta(t;c) = 0\,,
\end{align}
is then immediate. In equilibrium, the pricing rule is given by
\begin{align}
	\lambda(t;c) &= \frac{\beta(t;c)\,\Sigma(t;c)}{\sigma^2}\,,
\end{align}
and from \eqref{eqn:lambda_Sigma} we have that $\Sigma$ is a decreasing function of $t$ and therefore bounded by $\Sigma(0;c) = \Sigma_0^v$. Thus, we also have the limit
\begin{align}
	\lim_{c\rightarrow \infty} \lambda(t;c) &= 0\,,
\end{align}
uniformly in $t$.

For the remainder of the proof we only consider $A>0$, as the proof for $A=0$ is similar and more straightforward. The limits which correspond to $c\rightarrow 0$ rely on some of the results of the proof of Lemma \ref{lem:genExistAndUnique}. In particular, the solution to the FBODE \eqref{ode:FBODE} must satisfy
\begin{align*}
	\frac{dx_1}{dt}(t;c,k) &= f_1(x_1(t;c,k);c,k)\,, & x_1(0;c,k) &= \Sigma_0^v\,,\\
	f_1(x_1;c,k) &= -\frac{ 16\, A^2\, \sigma^2\, x_1^2\, \left(x_1^{\gamma(c)}  + k \right)^2 }{ \biggl((\gamma(c) +1)^2\, x_1^{\gamma(c)}  + (\gamma(c) -1)^2\, k\biggr)^2 }\,, 
\end{align*}
where
\begin{align*}
	\gamma(c) &= \sqrt{1 + 8\,A\,c\,\sigma^2}\,,
\end{align*}
ans where $k$ is chosen so that
\begin{align}
	(\gamma(c)+1)\,x_1^{\gamma(c)}(T;c,k) - (\gamma(c)-1)\,k = 0\,,\label{eqn:implicit}
\end{align}
and we have made all dependences on $c$ explicit. Note that all of these expressions above apply to both cases $c>0$ and $c=0$, with the provision that the choice of $k$ when $c=0$ is the positive root in \eqref{eqn:kpm}. Since the function $f_1$ is continuously differentiable with respect to $c$, so is the solution $x_1(t;c,k)$ (see \cite{hartman_ordinary_2002}). Additionally, $\gamma$ is continuously differentiable with respect to $c$, and so by the implicit function theorem we may take $k = k(c)$ in view of \eqref{eqn:implicit}. We then have $x_2$ given by
\begin{align*}
	x_2(t;c,k(c)) &= \rho(x_1(t;c,k(c))) = \frac{(\gamma(c)+1)\,x_1^{\gamma(c)}(t;c,k(c)) - (\gamma(c)-1)\,k(c)}{4\,A\,(x_1^{\gamma(c)}(t;c,k(c)) + k(c))}\,,
\end{align*}
which is also seen to be continuously differentiable with respect to $c$. Recall that in equilibrium the trading and pricing rules can be written as
\begin{align*}
	\beta(t;c) &= \frac{\sigma^2}{2\,(c\,\sigma^2 + x_2(t;c,k(c))}\,, & \lambda(t;c) &= \frac{x_1(t;c,k(c))}{2\,(c\,\sigma^2 + x_2(t;c,k(c))}\,.
\end{align*}
Both of these expressions are continuous functions of $t$ and $c$ except at $(t,c) = (T,0)$ where both denominators are equal to $0$. Thus, if we restrict $t$ to a compact subinterval on $[0,T)$ and $c$ to an interval of the form $[0,C]$, then both $\beta$ and $\lambda$ are uniformly continuous with respect to $(t,c)$, and the expressions given in \eqref{eq:betac0} and \eqref{eq:lambdac0} are obtained by direct substitution of $c=0$ using the corresponding solutions of $x_1$ and $x_2$ given in Lemma \ref{lem:genExistAndUnique}. \qed

\newpage

\section{Appendix B}

This appendix contains the full systems of equations that are solved in the proof of Proposition \ref{prop:single_auction_approx}.

\subsection{System for $\tilde{\lambda}_{i,j}$}

\begin{align*}
	\biggl(4\,\tilde{\lambda}_{0,0}^3 - 4\,\tilde{\lambda}_{0,0}\,\lambda_K^2\biggr) &= 0\,,\\
	\biggl(	4\,\tilde{\lambda}_{0,0}^4\,\sigma^2 - 4\,\tilde{\lambda}_{0,0}^2\,\lambda_K^2\,\sigma^2 + 12\,\tilde{\lambda}_{1,0}\,\tilde{\lambda}_{0,0}^2 - 4\,\tilde{\lambda}_{1,0}\,\lambda_K^2\biggr) &= 0\,,\\
	\biggl(	12\,\tilde{\lambda}_{0,0}^2\,\tilde{\lambda}_{0,1} - 4\,\tilde{\lambda}_{0,1}\,\lambda_K^2 + 8\,\tilde{\lambda}_{0,0}^2 - 8\,\lambda_K^2\biggr) &= 0\,,\\
	\biggl(\tilde{\lambda}_{0,0}^5\,\sigma^4 + 16\,\tilde{\lambda}_{0,0}^3\,\tilde{\lambda}_{1,0}\,\sigma^2 + 12\,\tilde{\lambda}_{2,0}\,\tilde{\lambda}_{0,0}^2 + 12\,\tilde\lambda_{0,0}\,\tilde\lambda_{1,0}^2 - 8\,\tilde\lambda_{0,0}\,\tilde\lambda_{1,0}\,\lambda_K^2\,\sigma^2 - 4\,\tilde\lambda_{2,0}\,\lambda_K^2\biggr) &= 0\,,\\
	\biggl(16\,\tilde\lambda_{0,0}\,\tilde\lambda_{1,0} + 12\,\tilde\lambda_{0,0}^2\,\tilde\lambda_{1,1} - 4\,\tilde\lambda_{1,1}\,\lambda_K^2 + 4\,\tilde\lambda_{0,0}^3\,\sigma^2 \hspace{30mm}\\
	+ 16\,\tilde\lambda_{0,0}^3\,\tilde\lambda_{0,1}\,\sigma^2 + 24\,\tilde\lambda_{0,0}\,\tilde\lambda_{0,1}\,\tilde\lambda_{1,0} - 8\,\tilde\lambda_{0,0}\,\tilde\lambda_{0,1}\,\lambda_K^2\,\sigma^2\biggr) &= 0\,,\\
	\biggl(12\,\tilde\lambda_{0,2}\,\tilde\lambda_{0,0}^2 + 12\,\tilde\lambda_{0,0}\,\tilde\lambda_{0,1}^2 + 16\,\tilde\lambda_{0,0}\,\tilde\lambda_{0,1} + 4\,\tilde\lambda_{0,0} - 4\,\tilde\lambda_{0,2}\,\lambda_K^2\biggr) &= 0\,,\\
	\biggl(5\,\tilde\lambda_{0,0}^4\,\tilde\lambda_{1,0}\,\sigma^4 + 16\,\tilde\lambda_{2,0}\,\tilde\lambda_{0,0}^3\,\sigma^2 + 24\,\tilde\lambda_{0,0}^2\,\tilde\lambda_{1,0}^2\,\sigma^2 + 12\,\tilde\lambda_{3,0}\,\tilde\lambda_{0,0}^2 + 24\,\tilde\lambda_{2,0}\,\tilde\lambda_{0,0}\,\tilde\lambda_{1,0} \hspace{20mm}\\
	- 8\,\tilde\lambda_{2,0}\,\tilde\lambda_{0,0}\,\lambda_K^2\,\sigma^2 + 4\,\tilde\lambda_{1,0}^3 - 4\,\tilde\lambda_{1,0}^2\,\lambda_K^2\,\sigma^2 - 4\,\tilde\lambda_{3,0}\,\lambda_K^2\biggr) &= 0\,,\\
	\biggl(16\,\tilde\lambda_{0,0}\,\tilde\lambda_{2,0} + 12\,\tilde\lambda_{0,1}\,\tilde\lambda_{1,0}^2 + 12\,\tilde\lambda_{0,0}^2\,\tilde\lambda_{2,1} - 4\,\tilde\lambda_{2,1}\,\lambda_K^2 + 8\,\tilde\lambda_{1,0}^2 + 5\,\tilde\lambda_{0,0}^4\,\tilde\lambda_{0,1}\,\sigma^4 \hspace{40mm}\\
	+ 12\,\tilde\lambda_{0,0}^2\,\tilde\lambda_{1,0}\,\sigma^2 + 16\,\tilde\lambda_{0,0}^3\,\tilde\lambda_{1,1}\,\sigma^2 + 24\,\tilde\lambda_{0,0}\,\tilde\lambda_{0,1}\,\tilde\lambda_{2,0} + 24\,\tilde\lambda_{0,0}\,\tilde\lambda_{1,0}\,\tilde\lambda_{1,1} \hspace{20mm}\\
	+ 48\,\tilde\lambda_{0,0}^2\,\tilde\lambda_{0,1}\,\tilde\lambda_{1,0}\,\sigma^2 - 8\,\tilde\lambda_{0,0}\,\tilde\lambda_{1,1}\,\lambda_K^2\,\sigma^2 - 8\,\tilde\lambda_{0,1}\,\tilde\lambda_{1,0}\,\lambda_K^2\,\sigma^2\biggr) &= 0\,,\\
	\biggl(4\,\tilde\lambda_{1,0} + 16\,\tilde\lambda_{0,0}\,\tilde\lambda_{1,1} + 16\,\tilde\lambda_{0,1}\,\tilde\lambda_{1,0} + 12\,\tilde\lambda_{0,1}^2\,\tilde\lambda_{1,0} + 12\,\tilde\lambda_{0,0}^2\,\tilde\lambda_{1,2} - 4\,\tilde\lambda_{1,2}\,\lambda_K^2 \hspace{40mm}\\
	+ 12\,\tilde\lambda_{0,0}^2\,\tilde\lambda_{0,1}\,\sigma^2 + 16\,\tilde\lambda_{0,0}^3\,\tilde\lambda_{0,2}\,\sigma^2 + 24\,\tilde\lambda_{0,0}\,\tilde\lambda_{0,1}\,\tilde\lambda_{1,1} + 24\,\tilde\lambda_{0,0}\,\tilde\lambda_{0,2}\,\tilde\lambda_{1,0} \hspace{10mm}\\
	+ 24\,\tilde\lambda_{0,0}^2\,\tilde\lambda_{0,1}^2\,\sigma^2 - 4\,\tilde\lambda_{0,1}^2\,\lambda_K^2\,\sigma^2 - 8\,\tilde\lambda_{0,0}\,\tilde\lambda_{0,2}\,\lambda_K^2\,\sigma^2\biggr) &= 0\,,\\
	\biggl(8\,\tilde\lambda_{0,3}\,\lambda_K^2 - 8\biggr) &= 0\,.
\end{align*}

\subsection{System for $\tilde{\beta}_{i,j}$}

\begin{align*}
	\biggl(16\,\tilde\beta_{0,0}^4\,\lambda_K^4 - 1\biggr) &= 0\,, \\
	\biggl(16\,\tilde\beta_{0,0}^3\,\lambda_K^4\,\sigma^2 + 64\,\tilde\beta_{1,0}\,\tilde\beta_{0,0}^3\,\lambda_K^4\biggr) &= 0\,, \\
	\biggl(32\,\tilde\beta_{0,0}^5\,\lambda_K^4 + 64\,\tilde\beta_{0,1}\,\tilde\beta_{0,0}^3\,\lambda_K^4 + 16\,\tilde\beta_{0,0}^3\,\lambda_K^2 + 2\,\tilde\beta_{0,0}\biggr) &= 0\,, \\
	\biggl(64\,\tilde\beta_{2,0}\,\tilde\beta_{0,0}^3\,\lambda_K^4 + 96\,\tilde\beta_{0,0}^2\,\tilde\beta_{1,0}^2\,\lambda_K^4 + 48\,\tilde\beta_{0,0}^2\,\tilde\beta_{1,0}\,\lambda_K^4\,\sigma^2\biggr) &= 0\,, \\
	\biggl(160\,\tilde\beta_{1,0}\,\tilde\beta_{0,0}^4\,\lambda_K^4 + 64\,\tilde\beta_{1,1}\,\tilde\beta_{0,0}^3\,\lambda_K^4 + 48\,\tilde\beta_{0,1}\,\tilde\beta_{0,0}^2\,\lambda_K^4\,\sigma^2 \hspace{20mm}\\
	+ 192\,\tilde\beta_{0,1}\,\tilde\beta_{1,0}\,\tilde\beta_{0,0}^2\,\lambda_K^4 + 48\,\tilde\beta_{1,0}\,\tilde\beta_{0,0}^2\,\lambda_K^2 + 2\,\tilde\beta_{1,0}\biggr) &= 0\,, \\
	\biggl(160\,\tilde\beta_{0,0}^4\,\tilde\beta_{0,1}\,\lambda_K^4 + 64\,\tilde\beta_{0,2}\,\tilde\beta_{0,0}^3\,\lambda_K^4 + 96\,\tilde\beta_{0,0}^2\,\tilde\beta_{0,1}^2\,\lambda_K^4 + 48\,\tilde\beta_{0,0}^2\,\tilde\beta_{0,1}\,\lambda_K^2 + 2\,\tilde\beta_{0,1}\biggr) &= 0\,.
\end{align*}

\section*{References}

\bibliographystyle{chicago}
\bibliography{bibfile}

\end{document}